\documentclass[12pt]{article}

\usepackage{amssymb,amsmath,amsfonts,eurosym,geometry,ulem,graphicx,caption,mathtools, xcolor,setspace,sectsty,comment,footmisc,caption,pdflscape,subfigure,array,hyperref, parskip, appendix}

\hypersetup{hidelinks}

\usepackage{natbib}

\makeatletter
\renewcommand\hyper@natlinkbreak[2]{#1}
\makeatother

\graphicspath{ {./images/} }

\DeclarePairedDelimiter\floor{\lfloor}{\rfloor}
\DeclarePairedDelimiter\ceil{\lceil}{\rceil}
\newtheorem{theorem}{Theorem}
\newtheorem{corollary}[theorem]{Corollary}
\newtheorem{proposition}{Proposition}
\newtheorem{lemma}[theorem]{Lemma}
\newenvironment{proof}[1][Proof]{\noindent\textbf{#1.} }{\hfill$\square$}

\newcolumntype{L}[1]{>{\raggedright\let\newline\\arraybackslash\hspace{0pt}}m{#1}}
\newcolumntype{C}[1]{>{\centering\let\newline\\arraybackslash\hspace{0pt}}m{#1}}
\newcolumntype{R}[1]{>{\raggedleft\let\newline\\arraybackslash\hspace{0pt}}m{#1}}

\begin{document}

\begin{titlepage}
\title{The importance of being discrete: on the (in-)accuracy of continuous approximations in auction theory\thanks{For comments and helpful suggestions, we would especially like to thank Miguel Ballester, Alan Beggs, Adam Brzezinski, Simon Cowan, Mikhail Drugov, Simon Finster, Caspar Jacobs, Bernhard Kasberger, Paul Klemperer, Nat Levine, Meg Meyer, Aiden Smith and Jasmine Theilgaard; as well as seminar audiences in Oxford and elsewhere.}}
\author{Itzhak Rasooly\thanks{Department of Economics, University of Oxford.} \and Carlos Gavidia-Calderon\thanks{Department of Computer Science, University College London.}}
\date{\today}
\maketitle
\vspace*{-0.4cm}
\begin{center}
\end{center}
\vspace*{-0.8cm}
\begin{abstract}
\noindent While auction theory views bids and valuations as continuous variables, real-world auctions are necessarily discrete. In this paper, we use a combination of analytical and computational methods to investigate whether incorporating discreteness substantially changes the predictions of auction theory, focusing on the case of uniformly distributed valuations so that our results bear on the majority of auction experiments. In some cases, we find that introducing discreteness changes little. For example, the first-price auction with two bidders and an even number of values has a symmetric equilibrium that closely resembles its continuous counterpart and converges to its continuous counterpart as the discretisation goes to zero. In others, however, we uncover discontinuity results. For instance, introducing an arbitrarily small amount of discreteness into the all-pay auction makes its symmetric, pure-strategy equilibrium disappear; and appears (based on computational experiments) to rob the game of pure-strategy equilibria altogether. These results raise questions about the continuity approximations on which auction theory is based and prompt a re-evaluation of the experimental literature.

\noindent \\
\vspace{-0.75cm}\\
\noindent \textsc{Keywords:} Auction Theory, Equilibrium, Discrete, Experiments
\vspace{0in}\\
\noindent\textsc{JEL Codes:} C70, C72, C91, D44\\

\bigskip
\end{abstract}
\setcounter{page}{0}
\thispagestyle{empty}
\end{titlepage}
\pagebreak \newpage


\section{Introduction} \label{introduction}

In auction theory, bids and valuations are typically viewed as continuous variables. That is to say, the theory assumes that players may bid any non-negative real number; and that valuations are continuously distributed over some subset of the real line. In actual auctions, however, both valuations and bids are discrete. For example, a participant in an auction would never be able to submit a bid of $\sqrt2$ pounds; and $\sqrt2$ pounds could never be the most that a participant is willing to pay for an object (i.e. their `valuation'). This raises a simple question: does introducing discreteness into standard game theoretic models of auctions substantially change their predictions?

Surprisingly, this question has received very little attention from economists. In the theoretical literature, there are some papers that examine the case of discrete valuations but continuous bids; and some others that examine the case of discrete bids but continuous valuations.\footnote{For discrete valuations and continuous bids, see  \cite{maskin1985auction}, \cite{riley1989expected}, \cite{wang1991common},  \cite{che2006revenue} and \cite{bergemann2017first}.
For discrete bids and continuous valuations, see  \cite{chwe1989discrete},
 \cite{rothkopf1994role}, \cite{david2007optimal}, \cite{cai2010note} and  \cite{goncalves2017note}.}
 However, there are far fewer that consider the case of discrete bids and discrete valuations. Moreover, those papers that do are interested in questions that are very different from ours. For instance, \cite{dekel2003rationalizable} introduce discrete values and bids in order to study the implications of rationalisability (whereas we focus on equilibrium); and in any case their results rest on the assumption that the number of bidders is very large.\footnote{There are also some papers which analyse the equilibria of auctions with discrete bids and known valuations \citep{boudreau2011all, li2017nash}.
 Our interest, however, lies in the case where valuations are uncertain.}

The question is also almost completely ignored in the experimental literature. In this literature, it is standard for experimenters to begin by writing down a continuous auction model and to then `test' the predictions of this model using a discrete laboratory experiment. It is occasionally asserted, without any real justification, that the discreteness makes little difference if the number of possible valuations is large (see, e.g., \cite{noussair2006behavior}, footnote 3). More often, however, the fact that discrete games are being used to test a continuous theory is not even discussed (see, e.g., \cite{coppinger1980incentives}; \cite{cox1982theory}; \cite{kagel1985individual};
\cite{dyer1989resolving}; \cite{chen1998nonlinear}; \cite{filiz2007auctions};
\cite{kirchkamp2011out}).\footnote{An important exception is a paper by \cite{goeree2002quantal}
that carefully analyses the equilibria of their experimental auction game. While their analysis is insightful, the case they consider (2 bidders, 6 possible valuations, 1 auction structure) is obviously a rather special one and could not (and is not intended to) shed very much light on the general case.} This is potentially concerning: if discrete equilibria are importantly different from continuous equilibria, then these experimental tests are invalid.


In this paper, we make substantial progress in answering this question through study of a discretised version of the canonical independent private values model. We assume that values and bids are restricted to the same finite and evenly spaced set (a fairly common set-up in the experimental literature, although not one that may be universally applied). We allow the discretisation to be arbitrarily fine (so that the difference between the bids can be as small as a penny). We then study the equilibria of three standard auction structures: the first-price, second-price and all-pay auctions. We also consider two different tie-breaking rules: fair tie-breaking (in which tying winners are chosen randomly) and the case in which players cannot win through a tie. While the former assumption is more popular in auction experiments, the latter is easier to study analytically. Ultimately, though, our results do not substantially hinge on which tie-breaking rule we assume.

We begin the analysis by considering the second-price auction. The central insight from classical auction theory -- that bidders should submit bids equal to their valuation -- applies in our discrete setting. Thus, introducing discreteness makes essentially no difference in the second-price auction, a fact which confirms the intuition of the many experimenters who have attempted to test the continuous theory using discrete laboratory auctions (\cite{coppinger1980incentives}; \cite{cox1982theory};
\cite{kagel1993independent}; \cite{aseff2004learning}; \cite{andreoni2007asymmetric};
\cite{cooper2008understanding}; \cite{drichoutis2015veil}; \cite{georganas2017optimistic}).

We then turn to the first-price auction and begin by proving some useful lemmas that hold regardless of the distribution that generates each player’s valuation. First, we show that bids are weakly increasing in values. Second, we show that symmetric equilibria cannot have jumps (for example, a situation in which bids jump from $3$ to $5$ when the valuation increases from $10$ to $11$). Finally, we show that if a discrete auction has a symmetric pure-strategy equilibrium, then it is (essentially) unique.

Using these results, we then study the equilibria of the first-price auction, focusing on the case of uniformly distributed valuations so that our results speak to the majority of auction experiments. In some cases (for example, with two bidders and an even number of values), we find that the auction can possess an equilibrium which is almost identical to the equilibrium in the continuous case and converges to the continuous equilibrium as the discreteness vanishes. In other cases, however, we uncover discontinuities. For example, with three or more bidders and a suitable tie breaking rule, the auction has no symmetric, pure-strategy equilibria -- in sharp contrast to the continuous case.

We then conduct a similar analysis of the all-pay auction. In this case, we once again find dramatic contrasts between our discrete games and their continuous counterparts (which are characterised by symmetric pure-strategy equilibria). Indeed, there are now no symmetric, pure-strategy equilibria at all -- a result that holds regardless of the number of bidders, tie-breaking rule and the fineness of the discretisation. Like our results for the first-price auction, this reveals a discontinuity in the set of equilibria (as the discretisation goes to 0) and contradicts assertions by experimenters about the equilibria of the games they study.

In principle, discrete auctions could have asymmetric equilibria that closely resemble the symmetric equilibria of their continuous counterparts. To investigate this possibility, we take a computational approach, writing down our games in normal form and computing their equilibria using standard solvers. Doing so is computationally challenging since the number of possible strategies quickly becomes very large (e.g. in one game whose equilibria we ultimately compute, each player has ten billion possible bidding functions to choose from). Nonetheless, we are able to significantly reduce the dimensionality of our problem through iterative deletion of dominated bids and thus compute the equilibria of a series of representative auction games.

In some cases (e.g. the first-price auction with 3 bidders), we do find asymmetric equilibria that closely resemble the symmetric equilibria of continuous theory. In general, however, the message of our computational results reinforces our analytical findings about symmetric equilibria. For almost all types of auction structure that we consider, we find examples in which the auction has no pure-strategy equilibria (symmetric or otherwise). This further reinforces the contrast between the discrete and the continuous case.

In cases where discrete auctions lack pure strategy equilibria, it follows from \cite{nash1950equilibrium} that their equilibria must be in mixed strategies. In our final section, we study such equilibria, paying particular attention to the question of whether the randomisation is `local' in the sense of occurring between at most two neighbouring bids. We show that this is emphatically \textit{not} the case: indeed, it is possible to construct discrete auctions in which one type randomises over nearly the entire range of feasible bids. Thus, not only can introducing discreteness make an interesting qualitative difference (by robbing auctions of pure strategy equilibria), but it can lead to equilibria which are quantitatively rather dissimilar from the pure strategy equilibria of textbook theory.

The remainder of this article is structured as follows. Section \ref{the-model} presents the (family of) models under consideration. Sections \ref{the-second-price-auction} -- \ref{asymmetric-equilibria} present the results for our three auction structures. Section \ref{sec:Mixed} studies mixed equilibria and shows that they can involve randomisation over an arbitrarily large range.  Finally, Section \ref{sec:Concluding remarks} concludes by discussing the implications of our findings for experimental and theoretical research.

\section{The model}
\label{the-model}

We consider $n \geq 2$ bidders who compete for a single and indivisible object. Bidder $i$'s value for this object $v_i$ is drawn from the finite set $\mathbb{X} = \{0, \delta, 2 \delta, ..., x\delta\}$ for some $\delta > 0$ and $x \in \mathbb{N}$. Without loss of generality, assume that every element in $\mathbb{X}$ is drawn with strictly positive probability; and (with some redundancy) let $S\equiv x+1$ denote the number of elements in $\mathbb{X}$. As usual, bidders know their own valuation but only the joint distribution of the values of the other bidders. Valuations are private and independently drawn. In addition, we assume that player $i$'s bid $b_i$ must belong to $\mathbb{X}$ (a common situation in the experimental literature, and one which plausibly applies to many non-experimental auctions).

In all auction structures that we consider, a player who submits a bid that is strictly higher than all opponent bids wins the auction. We consider two different rules that come into play in the event that two or more bids tie for highest. In the \textit{model without ties}, players win the auction if and only if their bid is strictly higher than all the bids submitted by their opponents. Thus, $\mathbb{P}(i \text{ wins}) = \mathbb{P}(b_i>b_j \hspace{0.5em} \forall j \neq i) $. In the \textit{model with ties}, a player can win if they tie with others who
also submit the highest bid. Specifically, if $m$ players submit the highest bid, they all have an equal chance ($1/m$) of being selected as the winner. This means, for example, that $\mathbb{P}(i\text{ wins}) = \mathbb{P}(b_i > b_j) + 0.5\mathbb{P}(b_i = b_j)$ in the case of $n = 2$.

Of course, bidder payoffs do depend on the auction structure. In the first-price auction, the player who submits the highest bid wins the object and pays a price equal to their bid. Thus, the expected payoff of a bidder with a value $v_i$ and submits a bid of $b_i$ is $\pi(v_i, b_i) = (v_i - b_i)\mathbb{P}(i\text{ wins})$. In the all-pay auction, players pay their bid even if they lose. Hence, $\pi(v_i, b_i) = v_i\mathbb{P}(i\text{ wins}) - b_i$. In the second-price auction, the player who submits the highest bid wins the object but pays a price equal to the second highest bid. We assume throughout that bidders are risk-neutral and therefore maximise their expected payoff.

Bidder's $i$'s (ex-ante) strategy is a function $\beta_i\colon \mathbb{X}\rightarrow \mathbb{X}$ mapping from each valuation they might have to a (permissible) bid. We define a (pure) \textit{equilibrium} as a set of bidding functions, one for each player, such that each player's bidding function maximises their expected payoff given the bidding functions of the other players.\footnote{We defer all discussion of mixed equilibria to Section \ref{sec:Mixed}.} (This is equivalent to a pure-strategy Bayes-Nash equilibrium with risk-neutral bidders.) When all bidders choose the same bidding function, we say that the equilibrium is \textit{symmetric}. We focus on symmetric, pure-strategy equilibria since these equilibria characterise the continuous case -- and the point of the exercise is to see how discrete equilibria compare to this benchmark.

Finally, we restrict attention to equilibria in undominated strategies (i.e. equilibria in which players never choose strategies which are weakly dominated). While this might seem restrictive, this cannot rule out any of the equilibria considered in the standard analysis of continuous auctions (since these are in undominated strategies). In addition, the restriction is easy to motivate: for example, players will never choose dominated strategies if their opponents `tremble' with arbitrarily low probability \citep{selten}. Since we only consider equilibria in undominated strategies, we omit the phrase ‘in undominated strategies’ when stating our results.
\vspace{-0.1cm}
\section{The second-price auction} \label{the-second-price-auction}

To begin with, we consider the second-price auction. Here the analysis is fairly trivial, as the following proposition indicates.\footnote{All proofs are collected in the appendix.}
\begin{proposition}\label{prop 1}
In the model with ties, each player sets $\beta(v) = v$ in the unique equilibrium. In the model without ties, the set of equilibria is precisely the set of strategy profiles in which each player sets either $\beta(v) = v$ or $\beta(v)=v+\delta$.
\end{proposition}
The point here is that in any equilibrium (in undominated strategies), players bid their valuation (or perhaps their valuation plus $\delta$). Of course, players also bid their valuation in the unique equilibrium in undominated strategies of the continuous model. Thus, introducing discreteness makes essentially no difference in the second-price auction, a fact that confirms the intuition of the many experimenters who have attempted to test the continuous theory using a discrete laboratory auction. Since the analysis of the second-price auction is so straightforward, we quickly move on to our next auction structure.
\vspace{-0.1cm}
\section{The first-price auction} \label{The First-Price Auction}
We begin our analysis of the first-price auction with some useful lemmas that hold independently of the distribution that generates the players’ valuations.
\begin{lemma}\label{lemma 1}
The set of equilibria does not depend on $\delta$.
\end{lemma}
Since our interest is in equilibria, Lemma \ref{lemma 1} allows us normalise $\delta = 1$ without loss of generality. In other words, we may view the bids and values as integers and strategies as a map from some set of consecutive integers to itself. This is not very surprising: obviously there is no important distinction between the game in which players bid $2\delta$ when their value is $3 \delta$, and the game in which they bid 2 when their value is 3. Holding the number of valuations and bids fixed, the only effect of varying $\delta$ is to induce a (positive) linear transformation in the payoffs. Throughout the following, then, we will set $\delta = 1$.
\begin{lemma}[No Dominated Strategies.]\label{lemma 2}
Suppose that a bidding function $\beta$ is chosen in an equilibrium. In the model with ties, $\beta(0)=0$ and $\beta(v) \leq v - 1$ for all $v \geq 1$. In the model without ties, $\beta(0)=0$, $\beta(1) \leq 1$ and $\beta(v) \leq v - 1$ for all $v \geq 2$.
\end{lemma}
Roughly speaking, Lemma \ref{lemma 2} rules out bids that exceed valuations: such bids are weakly dominated. As a result, it pins down bidding behaviour at the `start' of the bidding function (e.g., when $v = 0$) in a similar way to the boundary condition of continuous theory. As it will turn out, these low bids then pin down bids for higher valuations by inductive arguments. Thus, the assumption that players avoid dominated strategies will prove crucial for the analysis.
\begin{lemma}[Monotonicity.]\label{lemma 3}
Suppose that a bidding function $\beta$ is chosen in an equilibrium. Then for any two values $v, v' \in \mathbb{X}$, $v' > v$ implies that $\beta(v') \geq \beta(v)$.
\end{lemma}
Lemma \ref{lemma 3} states that equilibrium bids must be (weakly) increasing in valuations. While this must hold in equilibrium, it is perhaps worth noting that it holds more generally. To be more precise, suppose that the player believes that there is a strictly positive probability that they will win by submitting a bid of $\beta(v')$. Then, regardless of whether they expect their opponents to play equilibrium strategies, their optimal bid cannot decrease when their valuation increases from $v$ to $v'$. In other words, the monotonicity property arises more from the logic of best responses than from equilibrium reasoning.

We now examine some of the properties of \textit{symmetric} equilibria (henceforth, SE), which are characterised by a single bidding function $\beta$. To begin with, note that, in the model without ties, an SE cannot have jumps (i.e. a situation in which the bid jumps from $b$ to some $b' \geq b + 2$ when the valuation increases from $v$ to $v+1$). To see this, suppose (for contradiction) that there were such an SE and consider what would happen if a player with a value of $v+1$ were to deviate by reducing their bid to $b+1$. The deviation would not reduce their probability of winning: both before the deviation and after it, they win the auction if and only if the highest bid of their opponents is less than or equal to $b$. Obviously, however, it would increase their payoff in the event that they did win. Finally, since they had a strictly positive probability of winning in an SE (e.g. whenever their opponents all bid $b$), the deviation would strictly increase their expected payoff, contradicting the supposed optimality of their strategy in the candidate SE. We thus have the following result.
\begin{lemma}[No jumps]\label{lemma 4}
Suppose that $\beta$ is an SE in the model without ties. Then for every $v, v + 1 \in \mathbb{X}$, $\beta(v+1) \leq \beta(v)+1$.
\end{lemma}

It is not as straightforward to prove a corresponding lemma in a model with ties. The reason is that the reduction in a player’s bid from $b'$ to $b+1$ now does reduce the player’s probability of winning: for instance, it eliminates their chance of winning when their opponents all bid $b'$. Nonetheless, we are able to prove qualified versions of this lemma for particular distributions of player valuations (see Lemma \ref{lemma 13} for the details).

Finally, we prove a uniqueness result.
\begin{lemma}\label{lemma 5}
In the model without ties, the SE (if it exists) is unique up to a choice of $\beta(1)$.
\end{lemma}
Lemma \ref{lemma 5} states that, in the model without ties, the SE (if it exists) is uniquely implied by the choice of $\beta(1)$. Since players avoid dominated strategies, either $\beta(1) = 0$ or $\beta(1) = 1$. Hence, there are at most two SE in the first-price auction. In some cases, the choice of $\beta(1)$ makes no difference to higher bids, so the SE is essentially unique. In others, however, changing $\beta(1)$ has implications for bids further up the chain -- a fact that will be illustrated by the next proposition.

Using these results, we now examine the equilibria of the first-price auction. To do so, we will now assume that the values are uniformly distributed. The reasons are twofold. First, this distribution is by far most the popular choice by auction experimenters \citep{kagel1995auctions, kagel2011auctions};
so its analysis is of particular importance if we wish to interpret their results. Second, it turns out that it is impossible to say very much about the existence of SE for arbitrary distributions. Thus, some distributional choice is necessary.

We begin with the case of 2 bidders in the model without ties.
\begin{proposition}\label{prop 2}
In the first-price auction with two bidders and without ties, there are exactly two SE: $\beta(v)= \floor{v/2}$ and $\beta(v)= \ceil{v/2}$.
\end{proposition}
Proposition \ref{prop 2} demonstrates that continuous auction theory can be an extremely good approximation to discrete reality (see \cite{chwe1989discrete} for a similar result in the context of discrete bids but continuous valuations). In the continuous case, players set $\beta(v) = v/2$. In the discrete case, they also bid half their value when possible and they round their bid to one of the two closest integers otherwise. (In one SE, they round up; in the other, they round down.) Thus, the difference between continuous and discrete case can be extremely small (e.g. if the bids and values are discretised in pennies) and the difference vanishes as the discretisation goes to 0. That is,
\begin{corollary}\label{corollary 6}
In any SE with two bidders and without ties, bidding functions and expected revenues converge to their counterparts in the continuous auction as $\delta \rightarrow 0$.
\end{corollary}
Unfortunately, however, this happy state of affairs ceases to be once we introduce additional bidders (while maintaining our previous assumption about ties). In fact, we have the following rather striking result:
\begin{proposition}\label{prop 3}
Suppose that $n \geq 3$ and that the maximum $x$ satisfies
\begin{equation} \label{eq 1}
x>\frac{2^\frac{1}{n-1}}{2^\frac{1}{n-1}-1}.
\end{equation}
Then there are no SE in the first-price auction without ties.
\end{proposition}
To illustrate the argument, we sketch a proof for the case of $n = 3$. To begin, we study the start of the bidding function. By Lemma \ref{lemma 2}, $\beta(0) = 0$, $\beta(1) \leq 1$ , and $\beta(2) =1$. Furthermore, it is easy to check that $\beta(3) =2$ in any SE (if this were not true, i.e. $\beta(3) = 1$, it would then pay for a bidder with a value of 3 to deviate to a bid of 2.)

Of particular interest is the bid when $v=4$. By monotonicity and no jumps, there are only two possibilities: $\beta(4) =2$ and $\beta(4) = 3$. Let us consider these possibilities in turn. In the first instance, suppose that $\beta(4) =2$. In that case, a player with a value of 4 would win in equilibrium whenever both of their opponents’ values were 2 or lower. Since there are $3$ such values, their equilibrium (expected) pay-off would be $2(3/S)^2$ where $S$ is the number of possible valuations. However, if the player were to deviate by bidding 3, they would win (at least) whenever their opponents’ values were 4 or lower. Since there are 5 such values, their payoff would be at least $(5/S)^2> 2(3/S)^2$. Since the deviation would be strictly profitable, we cannot have $\beta(4) =2$ in an SE.

Let us then consider the second possibility: $\beta(4) =3$. In that case, a player with a value of $4$ would win in equilibrium whenever their opponents’ values were both 3 or lower. Since there are $4$ such values, their equilibrium pay-off would be $(4/S)^2$. However, if the player were to deviate by bidding $2$, they would win whenever their opponents’ values were both 2 or lower and therefore get an expected pay-off of $2(3/S)^2 > (4/S)^2$. Thus, we also cannot have $\beta(4) =3$ in an SE. Since these were the only 2 possibilities, this means that there is no SE in the case of $n = 3$.

We now briefly remark on the assumption about the number of possible values. As noted, the argument assumes that inequality (\ref{eq 1}) holds; i.e., that the number of valuations is large relative to the number of bidders. In our view, this is a reasonable requirement that is satisfied by virtually all real world and experimental auctions. If there are 1000 possible values (a fairly common strategy space in experiments), all this requires is that the number of bidders is fewer than 693. Even if the number of possible values is just 100, all this requires is that the number of bidders is fewer than 70. As an aside, we note that there \textit{is} an SE in the very high bidder case in which each player bids $0$ when $v=0$ and $\beta(v)=v-1$ otherwise. This last observation, incidentally, is consistent with the finding of \cite{dekel2003rationalizable} that an extremely large number of bidders can induce players to bid almost entirely their valuation.

It might be wondered whether this non-existence result survives as soon as ties are introduced. In fact, non-existence is common under either tie breaking rule (see the appendix for our full analysis of this case). Indeed, in some cases, allowing for ties can make things even worse. For example, in the case of 2 bidders (which had an SE in the no-ties case), SE now fail to exist if the number of possible values is odd.

In general, these results contradict assertions found in a number of experimental papers.
For example,  \cite{cox1982theory} experimentally simulate a first-price auction with two bidders, an odd number of values and fair tie-breaking; and assert (based on a lengthy discussion of the continuous case) that has a symmetric, pure strategy equilibrium (see their equation 3.5). As it turns out, however, their games have no pure SE at all --- and therefore cannot have the particular pure SE that they identify.

\vspace{-0.1cm}
\section{The all-pay auction} \label{the-all-pay-auction}
We now turn to the all-pay auction. As before, we begin with some useful lemmas.
\begin{lemma}\label{lemma 7}
The set of equilibria is invariant to $\delta$.
\end{lemma}
\begin{lemma}[No Dominated Strategies.]\label{lemma 8}
Suppose that a bidding function $\beta$ is chosen in an equilibrium. Then $\beta(0)=0$ and $\beta(v) \leq v - 1$ for all $v \geq 1$.
\end{lemma}
\begin{lemma}[Monotonicity.]\label{lemma 9}
Suppose that a bidding function $\beta$ is chosen in an equilibrium. Then for any two values $v, v' \in \mathbb{X}$, $v' > v$ implies that $\beta(v') \geq \beta(v)$.
\end{lemma}
\begin{lemma}[No jumps]\label{lemma 10}
Suppose that $\beta$ is an SE in the model without ties. Then for every $v, v + 1 \in \mathbb{X}$, $\beta(v+1) \leq \beta(v)+1$.
\end{lemma}
\begin{lemma}\label{lemma 11}
In the model without ties, the SE (if it exists) is unique.
\end{lemma}

In general, the logic behind these results is very similar to that in the first-price auction. We pause, however, to note one minor contrast. In the all-pay auction, restricting attention to undominated strategies implies that $\beta(0) =\beta(1)=0$. Thus there is no longer a choice to be made about the bid when $v=1$. When we combine this with the inductive argument used to establish Lemma \ref{lemma 5}, we see that the all-pay auction has at most one SE (in undominated strategies) while the first-price auction had at most two.

These preliminaries out of the way, we examine the existence of symmetric equilibria, again assuming uniformly distributed valuations so that our results speak to the majority of experimental auctions. In this case, we find an even more dramatic contrast between discrete auctions and their continuous counterparts:
\begin{proposition}\label{prop 4}
Suppose that $x \geq 10$. Then, regardless of the number of bidders and tie-breaking rule, there are no SE in the all-pay auction.
\end{proposition}

To (very partially) illustrate the logic of the argument, consider the case of the model without ties and let $n =2$. Where previously we studied the start of the bidding function, now we study the end. By Lemmas \ref{lemma 9} and \ref{lemma 10}, there are only a limited number of ways in which the bidding function can end. For example, it might be flat over the last three values or it might twice increase. It turns out, however, that no matter how it ends, no bidding function is a best response to itself. For example, if one bidder increased their bid twice, then the other should reduce their bid by 1 when their value is the maximum.

As before, Proposition \ref{prop 4} contradicts assertions by experimenters that discrete all-pay auctions are characterised by pure-strategy SE. For example, \cite{noussair2006behavior} conduct an experimental all-pay auction with $n=6$, 1000 possible bids/values and a fair tie breaking rule. While recognising that their game is discrete, they claim that it possesses the pure-strategy SE of continuous theory since the number of possible bids/values is large. In fact, we see that the situation is the opposite: it is precisely because the number of values is large that a pure-strategy SE does not exist.\footnote{In cases with a handful of possible bids and valuations, it is easy to verify that SE do exist. For instance, if $x = 2$, then (regardless of the number of bidders) there is an SE in which $\beta(0) = 0$, $\beta(1) = 0$ and $\beta(2) = 1$.} Moreover, Proposition \ref{prop 4} corrects their claim (see footnote 3) that discrete auction games differ from their continuous counterparts because in the discrete case there is a strictly positive probability of winning through a tie. In fact, our difference arises even in the model without ties (when winning through a tie is impossible).
\vspace{-0.1cm}
\section{Asymmetric equilibria} \label{asymmetric-equilibria}

In principle, discrete auctions could have asymmetric equilibria that closely resemble the symmetric equilibria of their continuous counterparts. To investigate this possibility, we take a computational approach, first representing the games in normal form and then finding their equilibria using standard solvers \citep{mckelvey2014gambit}. Since Bayesian games can be represented in sequence form, this approach is not computationally efficient \citep{von1996efficient}. However, the normal form representation allows us to apply standard algorithms which are guaranteed to find all pure-strategy equilibria and therefore establish the non-existence of pure-strategy equilibria in cases where the algorithmic search is unsuccessful.

Computing the equilibria of discrete auctions is challenging since the number of possible strategies is so large. Indeed, in a game with $k$ possible bids and valuations, each player has $k^k$ possible pure strategies (so the number of strategies grows super-exponentially). Nonetheless, we are able to significantly reduce the dimensionality of our problem through iterated deletion of dominated strategies. In the first stage, we delete strategies that involve bids above valuations since these strategies are weakly dominated. In subsequent stages, we iteratively delete strategies that are now strictly dominated.\footnote{For the details, see: https://github.com/cptanalatriste/discrete-auction-builder/tree/master/docs} In other words, we delete strategies using the Dekel-Fudenberg procedure \citep{dekel1990rational}. While not designed for this purpose, this procedure is guaranteed to preserve all equilibria in undominated strategies.\footnote{Clearly, the first round must preserve all equilibria in undominated strategies. Moreover, the subsequent rounds (in which strictly dominated strategies are deleted) are guaranteed to preserve all equilibria (in undominated strategies or otherwise).} Furthermore, it reduces the number of strategies quite dramatically and therefore allows us to construct payoff matrices with hundreds of thousands of entries where previously the number of entries would have run into the billions.

We then compute the equilibria of a variety of auction games. We vary the auction structure (first-price vs. all-pay), number of bidders (two vs. three), tie breaking rule and number of possible valuations. Throughout, we assume uniformly distributed values. In general, we avoid searching for equilibria in cases where we have already identified symmetric equilibria analytically (see Sections \ref{The First-Price Auction} -- \ref{the-all-pay-auction}).\footnote{Our code is available here: https://github.com/cptanalatriste/bayesian-game-builder}

In some cases, we do find asymmetric equilibria that closely resemble the symmetric equilibria of continuous theory. For example, consider a first-price auction with three bidders and no-ties; and suppose that the number of valuations is \textit{not} a multiple of three. Then (generalising one of our computational results) it easy to check that there is an asymmetric equilibrium in which players submit bids that are very close to the continuous equilibrium of $\beta(v) = 2v/3$ (see Figure 1 for an illustration).\footnote{Let $m$ denote an arbitrary multiple of $3$. In this equilibrium, the first bidder sets $\beta(0) = \beta(1) = 0$ and then $\beta(m) = 2m/3$, $\beta(m - 1) = \beta(m) - 1$ and $\beta(m + 1) = \beta(m) + 1$ (for all $m \geq 3$). The second bidder sets $\beta(0) = \beta(1) = 0$, $\beta(2) = 1$, $\beta(3) = \beta(4) = 2$; and then $\beta(m - 1) = \beta(m) = \beta(m+1) = 2m/3$. The third bidder sets $\beta(0) = \beta(1) = 0$, $\beta(2) = 1$, $\beta(3) = 2$; and then sets $\beta(m) = \beta(m - 1) = \beta(m - 2) = 2m/3 - 1$.} In this case, the discrete auction lacks symmetric equilibria (by Proposition \ref{prop 3}), but allowing for asymmetry provides a way of recovering continuous results.

\begin{figure}[h]
\centering
\begin{minipage}{0.8\textwidth} 
\caption{An asymmetric equilibrium with $n = 3$}
\vspace{-0.4cm}
\includegraphics[width = 13cm]{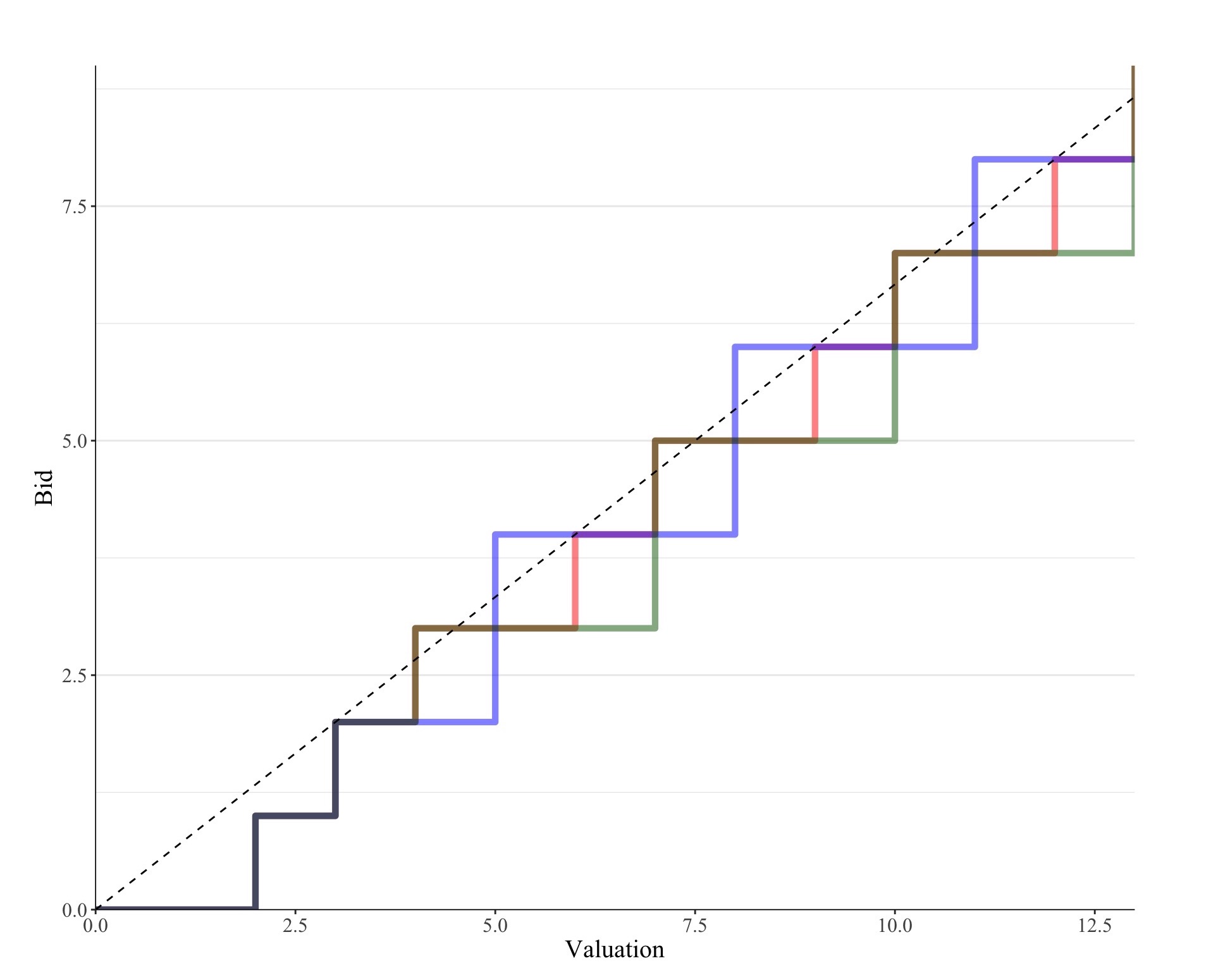}
\begin{minipage}{12.4cm}%
   \small \textit{Notes}: Figure 1 plots the asymmetric equilibrium discussed in the text for the case of $x = 13$. The dotted line depicts the continuous equilibrium of $\beta(v) = 2v/3$.%
\end{minipage}%
\end{minipage}
\end{figure}


In general, however, the analysis further emphasises the discontinuities highlighted in the previous section. Indeed, in six out of the seven types of auction game that we consider, we find examples in which the auction has no strategy equilibria whatsoever (see appendix Tables \ref{tab 1} and \ref{tab 2} for the full results). While we were only able to solve relatively small auctions, we conjecture (based on the results of Sections \ref{The First-Price Auction} -- \ref{the-all-pay-auction}) that non-existence would also be common for auctions with larger numbers of values.

\section{Mixed equilibria} \label{sec:Mixed}

The previous analysis demonstrates that discrete auctions can lack pure strategy equilibria, symmetric or otherwise. In such cases, it follows from \cite{nash1950equilibrium} that their only equilibria are in mixed strategies. It is natural to conjecture that, although these mixed equilibria differ qualitatively from the pure equilibria of textbook theory (insofar as they involve randomisation), they might match textbook theory quantitatively since the randomisation might be quite local. For example, if the textbook (continuous) theory predicts that a player should submit a bid of $6.5$, it is reasonable to guess that the discrete model with integer bids will predict randomisation between bids of 6 and 7. In this section, we show that this is emphatically \textit{not} the case. In fact, as the next result indicates, randomisation can take place over almost the entire feasible range of bids.

\begin{proposition}\label{mixed}
Consider the all-pay auction without ties. If the number of bidders is sufficiently large, players with a value of $x$ must submit all bids between $0$ and $x - 1$ with positive probability in any symmetric equilibrium.
\end{proposition}

\begin{figure}[h]
\centering
\begin{minipage}{0.8\textwidth} 
\caption{A mixed equilibrium with $n = 10$}\label{fig2}
\vspace{-0.4cm}
\hspace{-4em}\includegraphics[width = 16cm]{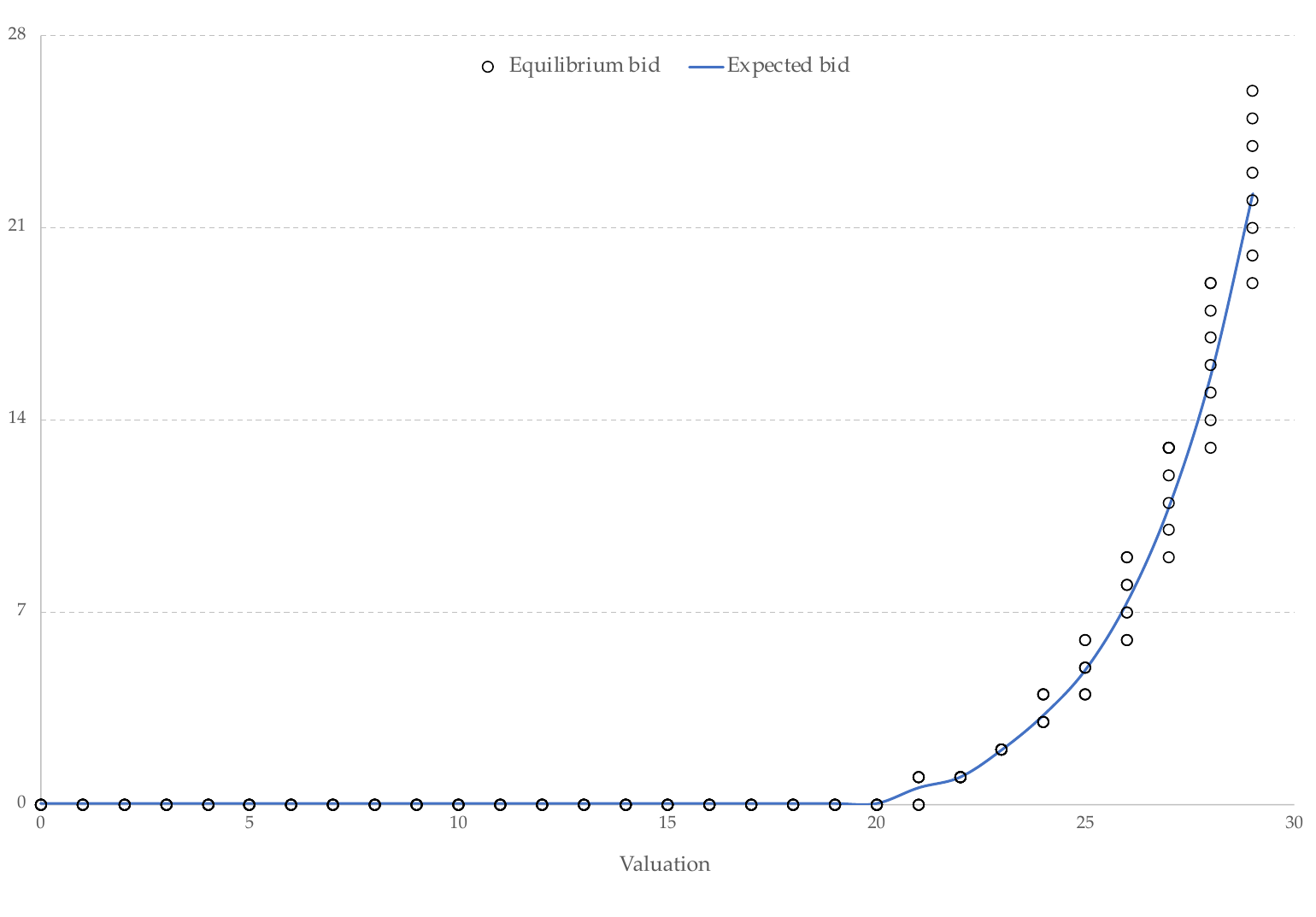}
\end{minipage}
\end{figure}

To understand the logic of Proposition \ref{mixed}, first consider players who do not have the maximum valuation, i.e. those for whom $v \leq x - 1$. As the number of players becomes large, such types become very unlikely to win in equilibrium since it becomes almost certain that one of their opponents will have drawn the maximum valuation. Thus, if the number of players is `sufficiently large' (a sufficient condition is implicit in our proof), such types must bid 0 in any symmetric equilibrium. Given the logic of Lemmas \ref{lemma 7}--\ref{lemma 10}, it then follows that players with the highest valuation will randomise over some interval of bids $0$, $1$, $2$, ..., up to some maximum bid $k$. If the number of players is large, such types will be tempted to win the auction for sure by `over-cutting' the maximum, i.e. submitting a bid of $k + 1$. The only way to prevent this (in the large $n$ limit) is by setting $k$ at $x-1$, which means that no profitable overcutting can occur.

\begin{figure}[h]
\centering
\begin{minipage}{0.8\textwidth} 
\caption{An equilibrium with normally distributed values}\label{fig3}
\vspace{-0.4cm}
\hspace{-4em}\includegraphics[width = 16cm]{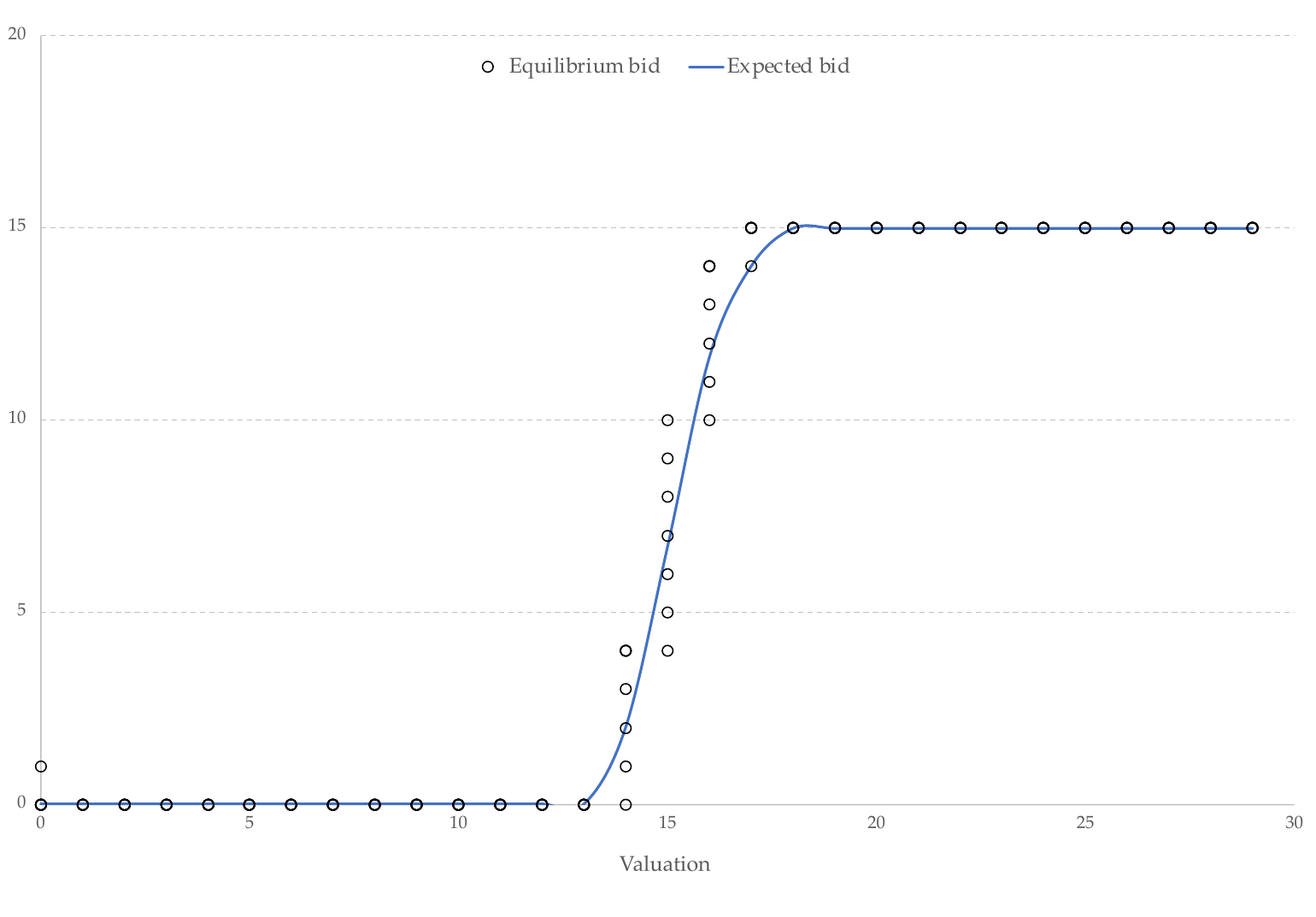}
\end{minipage}
\end{figure}

To illustrate the result, consider an all-pay auction with $30$ bids and valuations (i.e. a maximum of $x = 29$). One can then check that if $n \geq 101$, we do indeed obtain maximal randomisation in line with Proposition \ref{mixed}. Of course, the idea of 101 bidders may seem rather extreme; but it is important to emphasise that large amounts of randomisation can still be generated with fewer bidders. For example, if there are just 10 bidders, then players with the maximum valuation ($29$) will randomise between bids of $19$, $20$, ..., $26$. Meanwhile, players with a valuation of $28$ will randomise between bids of $13$, $14$, ... , $19$, and so forth (see Figure \ref{fig2} for illustration). Thus, randomisation need not be local; and one need not assume implausibly large numbers of bidders to obtain this result.

We should also emphasise that instead of assuming a large number of bidders, one could equally appeal to some substantial non-uniformity in the distribution of valuations. This should not come as a surprise: as far as the distribution of the maximum valuation goes, there is no difference between a model with two bidders with valuations drawn from distribution $F^{n/2}$, and a model with $n$ bidders with valuations drawn from distribution $F$. Moreover, the large scale randomisation need not arise at the maximum valuation. To illustrate this, consider an all-pay auction with just two bidders; and suppose that the valuations are normally distributed with a mean of $15$ and a standard deviation of $1$ (we choose this example since the normal distribution might be a sensible choice for applications). Such a distribution concentrates probability mass in the $14$ to $16$ range, so one would expect it to produce randomisation at such valuations. Consistent with this, we find that such types randomise over a number of bids: for example, players with a value of $15$ are predicted to randomise over all integers between $4$ and $10$ inclusive (see Figure \ref{fig3} for illustration). 

\section{Concluding remarks} \label{sec:Concluding remarks}
In this paper, we have made some progress towards understanding how auction theory changes once bids and valuations made discrete. In some cases, we find that introducing discreteness changes little. In others, however, we uncover discontinuity results. For example, in the all-pay auction with uniform values, introducing an arbitrarily small amount of discreteness makes the symmetric, pure-strategy equilibria of continuous theory disappear; and can make the pure-strategy equilibria disappear altogether. We have also shown that the remaining mixed equilibria can involve randomisation over the full range of non-dominated bids. In closing, we discuss the implications of these discontinuity results for theoretical and experimental research.

One conclusion for auction experimenters is that auction theory cannot be straightforwardly applied to the games they study. Indeed, ignoring the discreteness of laboratory games has led experimenters to make a series of errors (as documented in Sections \ref{The First-Price Auction} - \ref{the-all-pay-auction}). In future, experimenters may wish calculate the equilibria of their games directly, perhaps taking a computational approach (along the lines of Section \ref{asymmetric-equilibria}) or otherwise using our monotonicity, no-jumps and uniqueness lemmas to identify all \textit{symmetric} equilibria in auction games without ties.\footnote{This algorithm is implemented in a companion paper \cite{rasooly} and is available here: \url{https://auctionsolver.herokuapp.com/}}

For theorists, the implications are rather different. Over the last few decades, considerable effort has been expended in attempting to establish the existence of pure-strategy equilibria in successively more general auction models (see \cite{10.2307/2566961}, \cite{athey2001single} and
 \cite{reny2004existence} for notable examples). As our findings make clear, however, the existence of pure-strategy equilibria is far from guaranteed as soon as bids and valuations are made discrete. This is an issue since, as noted, discreteness is an inevitable feature of any real world or experimental auction and continuity approximations can be justified (only) as harmless assumptions that ease the analysis without substantially changing the results.

As discussed, the games we study do have equilibria: this is an immediate consequence of  \cite{harsanyi1967games} and \cite{nash1950equilibrium}. The upshot, however, is that these equilibria are often in mixed strategies. The question then arises of whether it is plausible to suppose that players really do randomise when bidding in auctions (and in just the way required by auction theory) even though they have no strict incentive to do so. While we cannot address this question in full, we briefly note three points that may be relevant.

First, the `classical' justification of mixed strategies --- that are necessary to keep opponents guessing about what a player will do \citep{von1947theory} ---  would appear to have even less force than usual in this context (see \cite{rubinstein1991comments} for a general criticism of this type of justification). In auction games with plausible type spaces, individuals will already have very little idea about their opponent’s valuation. Thus, even if they somehow knew their opponents’ bidding functions, they would already have substantial uncertainty about their opponents’ bids. That is to say, it is not clear that people would ever feel the need to submit random bids in order to `keep their opponents guessing'.

Second, even if individuals do choose to randomise in some environments --- as suggested, for instance, by \cite{chiappori2002testing} --- it is not clear that they do so in auction contexts. Of course this is an empirical question and one that should be tested by empirical studies of auction behaviour. Nonetheless, we conjecture that the results of such studies will find little evidence of random bidding by auction participants.

Third, introducing randomisation would appear to increase the cognitive requirements presupposed by equilibrium models. In a pure-strategy Bayes-Nash equilibrium, players are assumed to (i) know the joint distribution of opponent valuations, (ii) know which bidding functions their opponents have chosen, (iii) infer the joint distribution of opponent bids, (iv) infer the distribution of the highest opponent bid, and finally (v) calculate the optimal bid taking this distribution as given.\footnote{Technically speaking, these steps are sufficient but not necessary for individuals to play a Bayes-Nash equilibrium. For instance, they might hold incorrect beliefs about their opponents' decision rules, fail to optimise correctly and then submit equilibrium bids because these mistakes end up cancelling out. In practice, however, it is very hard to see why anyone would ever play a Bayes-Nash equilibrium without proceeding through these steps.} Once we introduce randomisation, this already difficult task becomes even harder. For now, players need to combine their opponents' probability distribution over bidding functions with their probability distribution over valuations in order to arrive at the distribution over bids (in addition to all the other steps). In other words, players now best respond to two sources of randomness: random valuations and random bidding functions. It is unclear whether we can expect this from even relatively sophisticated auction participants.

It seems, then, that our results should push one in the direction of greater scepticism towards auction theory. Less controversially, however, they highlight the need for further research on discrete auctions. Despite the large number of papers on auction theory, the world of discrete auctions remains almost completely unexplored. We hope that the results of this paper will stimulate and provide the basis for further work in this area.

\clearpage

\begin{appendix}

\setlength{\bibhang}{0pt}
\bibliographystyle{apalike}
\bibliography{bibliography.bib}


\newpage

\section{Tables} \label{sec:tab}

The following tables summarise our computational results on the existence of pure-strategy Bayes-Nash equilibria. The first column (`Valuations') specifies the number of possible valuations and bids in the game. The subsequent columns specify the auction structure and number of bidders. An entry of `yes' means that a pure-strategy Bayes-Nash equilibrium exists; an entry of `no' means that it does not. A blank entry means that we did not search for the equilibrium (either because we had around found symmetric equilibria analytically or because the payoff matrices were too large for this to be computationally feasible).

\begin{table}[h]
\centering
\caption{The existence of pure-strategy equilibria in the model with ties}\label{tab 1}
\label{tab:table1}
\begin{tabular}{|lllll|}
\hline
Valuations & FPSB, $n = 2$ & FPSB, $n = 3$ & All-pay, $n = 2$ & All-pay, $n = 3$ \\ \hline
6 & - & Yes & Yes & No \\
7 & - & No & No & No \\
8 & - & Yes & No & - \\
9 & - & - & Yes & - \\
10 & No & - & - & - \\ \hline
\end{tabular}
\end{table}

\begin{table}[h]
\centering
\caption{The existence of pure-strategy equilibria in the model without ties}\label{tab 2}
\label{tab:table2}
\begin{tabular}{|lllll|}
\hline
Valuations & FPSB, $n = 2$ & FPSB, $n = 3$ & All-pay, $n = 2$ & All-pay, $n = 3$ \\ \hline
6          & -           & Yes         & Yes            & Yes            \\
7          & -           & Yes         & No             & No             \\
8          & -           & -           & No             & -              \\
9          & -           & -           & Yes            & -              \\
10         & -           & -           & -              & -              \\ \hline
\end{tabular}
\end{table}

\clearpage




\section{Proofs} \label{sec:appendixa}

\begin{proof}[Proof of Proposition \ref{prop 1}]
In the model with ties, a standard argument establishes that the bidding function $\beta(v) = v$ weakly dominates all other bidding functions. Thus, each player sets $\beta(v) = v$ in the unique equilibrium in undominated strategies.

In the model without ties, the strategies $\beta(v) = v$ and $\beta(v) = v + \delta$ necessarily yield the same payoff (the only difference is that bidding $v + \delta$ allows a player to win when the highest bid of their opponents is $v$; but in that case the payoff from winning is zero). Furthermore, a standard argument establishes that $\beta(v) = v$ and $\beta(v)=v+\delta$ weakly dominate all other bidding functions. Thus, the set of equilibria in undominated strategies is precisely the set of strategy profiles in which each player sets either $\beta(v) = v$ or $\beta(v)=v+\delta$.\end{proof}

\begin{proof}[Proof of Lemma \ref{lemma 1}]
Consider the game in which bids and values are restricted to $\mathbb{X} = \{0, \delta, 2\delta, ..., x\delta\}$ for some $\delta > 0$. Thus, each player $i$ has a value of $\delta v_i$ for some $v_i \in \{0, 1, ..., x\} \equiv \mathbb{X}' $ and must choose to bid $\delta b_i$ for some $b_i \in \mathbb{X}'$. Consider initially the model without ties. In equilibrium, each player chooses their bid $\delta b_i$ so that, for all $b_i' \in \mathbb{X}'$,
\begin{equation}
\begin{split}
& (\delta v_i - \delta b_i)\mathbb{P}\text{(win}|\delta b_i) \geq (\delta v_i - \delta b_i')\mathbb{P}\text{(win}|\delta b_i') \\ & \iff
(v_i - b_i)\mathbb{P}\text{(win}|\delta b_i) \geq (v_i - b_i')\mathbb{P}\text{(win}|\delta b_i')
\\ & \iff
(v_i - b_i)\mathbb{P}(\delta b_i > \delta b_j \hspace{0.2em} \forall j \neq i) \geq (v_i - b_i')\mathbb{P}(\delta b_i' > \delta b_j \hspace{0.2em} \forall j \neq i)
\\ & \iff
(v_i - b_i)\mathbb{P}(b_i > b_j \hspace{0.2em} \forall j \neq i) \geq (v_i - b_i')\mathbb{P}(b_i' > b_j \hspace{0.2em} \forall j \neq i).
\end{split}
\end{equation}
Clearly, this inequality does not depend on $\delta$. Thus, if there is an equilibrium in the game with $x$ values and discretisation $\delta > 0$ in which player $i$ bids $b\delta$ when her valuation is $v\delta$, then there is an analogous equilibrium in the game with $x$ values and arbitrary discretisation $\delta' > 0$ in which player $i$ bids $b\delta'$ when her valuation is $v\delta'$. A similar argument establishes that the set of equilibria is invariant to $\delta$ in the model with ties.\end{proof}

\begin{proof}[Proof of Lemma \ref{lemma 2}]
Consider first the model with ties. For all $v$, bidding $b>v$ is weakly dominated by $b=0$: bidding $b = 0$ ensures a non-negative payoff, whereas bidding $b > v$ leads either to a zero payoff (if the player loses) or a negative payoff (if the player wins). Hence, $\beta(v) \leq v$ in any equilibrium in undominated strategies; and therefore $\beta(0) = 0$. Moreover, for any $v \geq 1$, bidding $b = 0$ also weakly dominates bidding $b = v$ (since bidding $b = 0$ can lead to a strictly positive payoff if all other players bid $0$). Thus, $\beta(v) \leq v - 1$ for all $v \geq 1$ in any equilibrium in undominated strategies.

Similar arguments apply in the model without ties. As before, bidding $b>v$ is weakly dominated by $b=0$: the latter ensures a payoff of zero, whereas the former yields either a negative or zero payoff. Hence, $\beta(v) \leq v$ for all $v$ in any equilibrium in undominated strategies; which implies that $\beta(0) = 0$ and $\beta(1) \leq 1$. Suppose now that $v \geq 2$. If a player bids $v$, their payoff is necessarily zero. However, if they bid $1$, their payoff is either zero (if their opponents all bid $1$ or more) or positive (if their opponents all bid $0$). Since $\beta(0) = 0$, there is a strictly positive probability that all opponents bid zero. Hence, bidding $b = 1$ now strictly dominates $b = v$, which implies that $\beta(v) \leq v-1$ for $v \geq 2$.\end{proof}

\begin{proof}[Proof of Lemma \ref{lemma 3}]
See \cite{goeree2002quantal} Proposition A1.\end{proof}

\begin{proof}[Proof of Lemma \ref{lemma 4}]
See the main text.\end{proof}

\begin{proof}[Proof of Lemma \ref{lemma 5}]
Since players avoid dominated bids, $\beta(0) = 0$ and either $\beta(1) = 0$ or $\beta(1) = 1$. Moreover, it is obvious that $\beta(2) = 1$ in any SE. We will now argue that the SE bidding function up to and including an arbitrary valuation $v \geq 2$ uniquely determines the SE bidding function up to and including valuation $v + 1$. By induction, the SE is therefore uniquely implied by our choice of $\beta(1)$.

To do this, let $b$ denote the bid submitted when the value is $v \geq 2$. By monotonicity and no-jumps, there are only 2 possibilities when the value is $v+1$:
\begin{enumerate}
    \item $\beta(v+1) = b$
    \item $\beta(v+1) = b + 1$
\end{enumerate}
Let us examine these possibilities in turn.

First possibility: $\beta(v+1)=b$. Then bidding $b$ with a value of $v+1$ must be at least as good as bidding $b+1$, i.e.
\begin{equation}
\underbrace{(v + 1 - b)\mathbb{P}( \leq b -1)^{n-1}}_A \geq \underbrace{(v - b)\mathbb{P}( \leq b)^{n-1}}_B
\end{equation}
where $\text{P(} \leq k)$ denotes the equilibrium probability that a player bids $k$ or lower.

Second possibility: $\beta(v+1) = b+1$. Then bidding $b+1$ with a value of $v+1$ must be at least as good as bidding $b$, i.e.
\begin{equation}
\underbrace{(v + 1 - b)\mathbb{P}'( \leq b -1)^{n-1}}_{A'} \leq \underbrace{(v - b)\mathbb{P}'( \leq b)^{n-1}}_{B'}
\end{equation}
Now notice that, in both cases, the equilibrium probability of submitting a bid that is $b-1$ or less is the same. That is, $\mathbb{P}( \leq b -1) = \mathbb{P}'( \leq b -1)$ which in turn implies that $A = A'$.

On the other hand, the equilibrium probability of submitting a bid that is $b$ or less does vary across the scenarios. In the first scenario, a player with a value of $v+1$ always bids $b$. In the second, they never bid $b$ and (by monotonicity) neither do players with higher values. Therefore, $ \mathbb{P}( \leq b) > \mathbb{P}'( \leq b)$. Assuming that $v - b > 0$ --- which is without loss of generality since $v \geq 2$ --- this means that $ B > B'$.

Let us now recall our two conditions (and substitute $A = A'$):
\begin{enumerate}
    \item $A \geq B$
    \item $A \leq B'$
\end{enumerate}
Since $ B > B'$, the truth of either condition implies the falsity of the other. Hence, at most one can hold; which means that $\beta(v+1)$ is uniquely implied by $\beta(v)$. By induction, this then means that the SE (if it exists) is uniquely implied by the choice of $\beta(1)$.\end{proof}

\begin{proof}[Proof of Proposition \ref{prop 2}]
To demonstrate existence, we will show that if one player (say player $2$) bids according to one of the two equilibrium strategies, then the other (player $1$) can do no better than follow the same strategy. For simplicity, assume that $x$ is odd (a similar argument applies when $x$ is even and is therefore omitted). If player 2 bids $b_2 = \floor{v_2/2}$ for all $v_2 \in \mathbb{X}$, and therefore never bids more than $\floor{x/2} = (x-1)/2$, then for player 1
\begin{equation}
\begin{split}
\mathbb{P}\text{(win}|b_1) = \mathbb{P}(b_1 > b_2) = 
\begin{cases}
1 & \text{if  } b_1 >(x-1)/2 \\
2b_1/(x + 1) & \text{if  } b_1 \leq (x-1)/2
\end{cases}
\end{split}
\end{equation}
We claim that no bid $b_1 >(x-1)/2$ is ever strictly better than bidding $b_1 = \floor{v/2}$. To see this, it is enough to consider the best such bid, namely $b_1 = (x + 1)/2$; and to show that even the highest type ($v_1 = x)$ would not strictly gain from deviating to such a bid. By bidding $b_1 = (x + 1)/2$ (and winning for sure), such a type would get a pay-off of
\begin{equation}
\pi(v_1 = x, b_1 = (x + 1)/2) = x - \frac{x+1}{2} = \frac{x-1}{2}
\end{equation}
Meanwhile, bidding the candidate SE would yield
\begin{equation}
\pi(v_1 = x, b_1 = \floor{x/2}) = \left(x - \frac{x-1}{2} \right) \mathbb{P}\text{(win}|b_1) =  \left(\frac{x+1}{2}\right)\left(\frac{x-1}{x+1}\right) = \frac{x-1}{2}
\end{equation}
Thus, even the $v = x$ type cannot strictly gain by bidding $b_1 > (x-1)/2$; which means that it is optimal for all types $v \in \mathbb{X}$ to bid $b_1 \leq (x-1)/2$.

Let us then consider bids in the range $b_1 \leq (x-1)/2$; and let us write player $1$'s bid as $b_1 = \floor{v_1/2} + i$ for some $i \in \mathbb{N}$. Given that $b_2 = \floor{v_2/2}$, player 1's expected payoff is
\begin{equation}
\begin{split}
\pi(v_1, b_1) &= (v_1 - b_1)\mathbb{P}(b_1 > b_2) \\
&= (v_1 - \floor{v_1/2}-i)\mathbb{P}(\floor{v_1/2} + i > \floor{v_2/2}) 
\end{split}
\end{equation}
If $v_1$ is even, then $\floor{v_1/2} = v_1/2$, so
\begin{equation}
\begin{split}
\pi(v_1, b_1) &= (v_1 - v_1/2 - i)\mathbb{P}(v_1/2 + i > \floor{v_2/2}) \\
&= \left(\frac{v_1}{2} - i\right)\left(\frac{v_1 + 2i}{x + 1}\right) \\
&= \frac{v_1^2}{2(x+1)} - \frac{2i^2}{x + 1}
\end{split}
\end{equation}
which is maximised at $i = 0$. Thus, $b_1 = \floor{v_1/2}$ is optimal for even $v_1$. Meanwhile, if $v_1$ is odd, then $\floor{v_1/2} = (v_1-1)/2$, so
\begin{equation}
\begin{split}
\pi(v_1, b_1) &= (v_1 - (v_1-1)/2 - i)\mathbb{P}((v_1-1)/2 + i > \floor{v_2/2}) \\
&= \left(\frac{v_1+1}{2} - i\right)\left(\frac{v_1 - 1 + 2i}{x + 1}\right) \\
&= \frac{(v_1-1)(v_1 + 1)}{2(x+1)} + \frac{2i(1-i)}{x + 1}
\end{split}
\end{equation}
which is again maximised at $i = 0$. Thus, $\beta_1(v_1)=\floor{v_1/2}$ is a best response for both odd and even $v_1$; which means that $\beta(v)=\floor{v/2}$ is a pure SE. A similar argument establishes the existence of an SE in which each player sets $\beta(v)=\ceil{v/2}$. Finally, we see that these are the \textit{only} pure SE by applying Lemma \ref{lemma 5}.\end{proof}

\begin{proof}[Proof of Corollary \ref{corollary 6}]
In the continuous case, the unique symmetric equilibrium is given by $\beta(v)=v/2$. In the discrete case with discretisation $\delta$, the SE are $\beta(\delta v)= \floor{v/2}\delta$ and $\beta(\delta v)= \ceil{v/2}\delta$ for all values $\delta v \in \mathbb{X}$ (i.e. for $v = 1, 2, ..., k$). Thus, the difference between a discrete bid and continuous bid is at most $\delta/2 \rightarrow 0$ as $\delta \rightarrow 0$ (holding the valuation $\delta v$ fixed). This establishes the (uniform) convergence of the discrete bidding function.

We now sketch a proof for the convergence of expected revenues. As $\delta \rightarrow 0$, each player's valuation converges in distribution to the valuation in the continuous case (i.e. a continuously and uniformly distributed random variable on the interval $ [ 0, x ] $). By the continuous mapping theorem, each player's bid converges in distribution to the bid in the continuous case (i.e. a continuously and uniformly distributed random variable on $[0, x/2]$). Applying the continuous mapping theorem again, we see that the maximum bid also converges in distribution to its continuous counterpart. Finally, by the bounded convergence theorem, the expected value of the maximum bid also converges to its continuous analogue.\end{proof}

\begin{proof}[Proof of Proposition \ref{prop 3}]
Let us suppose (for contradiction) that there were an SE and begin by studying the start of the bidding function that defines it. As before, $\beta(0)=0$ and $\beta(1)=0$ or $\beta(1)=1$ (since players avoid dominated strategies). Either way, $\beta(2) = 1$ (since bidding 0 or 2 yield a payoff of 0 whereas bidding 1 yields a positive payoff).



Let us now consider the first value at which the bidding function does not increase. To this end, define $k$ as the smallest value such that $\beta(k) = \beta(k-1)$ and $k \geq 2$. First, we argue that $k$ must exist. For contradiction, suppose that it does not, i.e. $\beta(v) = v - 1$ for all $v \geq 2$. In equilibrium, a player with $v=x$ would win whenever all of their opponents have values of $x - 1$ or lower. Since there are $x$ such values, and their payoff in the event of winning is $1$, their equilibrium expected payoff would be
 \begin{equation}
 \left( \frac{x}{x+1} \right)^{n-1}
 \end{equation}
If such a player were to bid one less, they would obtain an expected payoff of
\begin{equation}2\left( \frac{x-1}{x+1} \right)^{n-1}
\end{equation}
which (rearranging) is strictly larger if and only if
\begin{equation}
x>\frac{2^\frac{1}{n-1}}{2^\frac{1}{n-1}-1}.
\end{equation}
But this is precisely what we have assumed. Thus, if there is an SE, and a player with a value of $x$ does \textit{not} want to deviate, then there must be some $k \geq 2$ such that $\beta(k) = \beta(k-1)$.

Next, we locate $k$. In equilibrium, a player with a value $v = k$ bids $v - 2$, thereby winning whenever their opponents all have values of $v - 2$ or lower. Since there are $v - 1$ such values, the equilibrium expected payoff of a player with a value $v = k$ is
\begin{equation}
2 \left( \frac{v-1}{x+1} \right)^{n-1}
\end{equation}
If the player were to bid $1$ more (i.e. were to bid $v - 1$), they would win (at least) whenever their opponents all have values of $v$ or lower. Thus, their expected payoff would be at least
\begin{equation}
\left( \frac{v+1}{x+1} \right)^{n-1}
\end{equation}
Since an equilibrium cannot have strictly profitable deviations,
\begin{equation}\label{eq 2}
2 \left( \frac{v-1}{x+1} \right)^{n-1} \geq \left( \frac{v+1}{x+1} \right)^{n-1}
\end{equation}
We claim that $v = k$ is the smallest integer that satisfies inequality (\ref{eq 2}). To see this, notice that for all values smaller than $k$, submitting the equilibrium bid $\beta(v) = v - 1$ must be at least as good as bidding $v - 2$. This means that, for any $v < k$,
\begin{equation}
2 \left(\frac{v-1}{x+1}\right)^{n-1} \leq \left(\frac{v}{x + 1}\right)^{n-1} < \left(\frac{v+1}{x + 1}\right)^{n-1}
\end{equation}
so such values cannot satisfy (\ref{eq 2}).

We can re-write (\ref{eq 2}) as
\begin{equation}
v \geq \frac{1 + 2^\frac{1}{n-1}}{2^\frac{1}{n-1}-1}. \end{equation}
Since $k$ is the smallest integer that satisfies this inequality,
\begin{equation}
k = \frac{1 + 2^\frac{1}{n-1}}{2^\frac{1}{n-1}-1} + \epsilon
\end{equation}
for some $\epsilon \in [0, 1)$.

Having located $k$, we now consider a bidder with a value $v = k-1$. By assumption, they prefer to bid $\beta(k-1) = k-2$ than to bid 1 less, i.e.
\begin{equation}
\left( \frac{k-1}{x+1} \right)^{n-1} \geq 2\left( \frac{k-2}{x+1} \right)^{n-1}
\end{equation}
Equivalently,
\begin{equation}
k \leq \frac{2 (2^{\frac{1}{n-1}}) - 1}{2^{\frac{1}{n-1}} - 1}
\end{equation}
Inserting our value for $k$, we find that
\begin{equation}
\epsilon \leq \frac{2^{\frac{1}{n-1}}-2}{2^{\frac{1}{n-1}}-1}<0
\end{equation}
which contradicts $\epsilon \geq 0$, implying that there is no SE.
\end{proof}

\begin{proof}[Proof of Lemmas \ref{lemma 7} -- \ref{lemma 11}]
The arguments for Lemmas \ref{lemma 7}, \ref{lemma 8}, \ref{lemma 10} and \ref{lemma 11} are straightforward variants on those that establish analogous lemmas in the first-price auction. Lemma \ref{lemma 9} (monotonicity) requires a slightly different proof, something we now provide. First, since players avoid dominated bids, $\beta(0) = 0$. Thus, bids cannot decrease when the value increases from $v = 0$ to $v=1$. Next, consider now the bidding function starting from $v>0$ and suppose (for contradiction) that
\begin{enumerate}
    \item The optimal bid is $B$ when the value is $v>0$.
    \item The optimal bid is $b < B$ when the value is $V > v$.
\end{enumerate}
From (1), bidding $B$ must be better than bidding $b$ with a value of $v$:
\begin{equation}
\mathbb{P}\text{(win}|B)v - B \geq \mathbb{P}\text{(win}|b)v - b
\end{equation}
which (since $v > 0$) implies that
\begin{equation}
\mathbb{P}\text{(win}|B) - \mathbb{P}\text{(win}|b) \geq \frac{B-b}{v}.
\end{equation}
From (2), bidding $b$ must be better than bidding $B$ with a value of $V$:
\begin{equation}
\mathbb{P}\text{(win}|b)V - b \geq \mathbb{P}\text{(win}|B)V - B
\end{equation}
which (since $V > 0$ and $V > v$) implies that
\begin{equation}
\mathbb{P}\text{(win}|B) - \mathbb{P}\text{(win}|b) \leq \frac{B-b}{V} < \frac{B-b}{v}
\end{equation}
which is a contradiction.\end{proof}

\begin{proof}[Proof of Proposition \ref{prop 4}]
To begin with, consider the model without ties. As a preliminary, let us note that $\beta(x) \leq x - 2$. The argument is straightforward. Since players avoid dominated bids, $\beta(0) = \beta(1) = 0$. Furthermore, it is easy to check that $\beta(2) = 0$ for all $n\geq 2$. Otherwise, we could have $\beta(2) = 1$, but then the expected payoff from equilibrium bidding when $v=2$ would be
\begin{equation}
2\left( \frac{2}{x+1} \right)^{n-1} - 1 \leq 2\left( \frac{2}{5} \right)^{n-1} - 1 \leq 2\left( \frac{2}{5} \right) - 1< 0
\end{equation}
where the first inequality uses $x + 1 \geq 5$ and the second uses $n \geq 2$. Since the expected payoff is negative, the player could do better by bidding $0$, confirming that $\beta(2) \neq 1$ and therefore $\beta(2) = 0$. By no-jumps, this means that $\beta(v) \leq v - 2$ for all $v \geq 2$ and in particular for $v = x$.

This result in hand, we now consider two separate cases: $n=2$ and $n \geq 3$.

First case: $n=2$. We will consider the three highest bids $\beta(x)$, $\beta(x-1)$ and $\beta(x-2)$, i.e. the ending of the bidding function. We will argue that none of the possible endings can form part of a an SE.

Step 1. We cannot have two or more identical bids at the end, i.e. $\beta(x)=\beta(x-1)$.

Suppose (for contradiction) that there were two or more such bids at the end and consider a player with $v=x$. If the player increased their bid by $1$, the cost would be $1$. Meanwhile, the benefit would be
\begin{equation}
v \Delta \mathbb{P} \geq x \left( \frac{2}{x+1} \right) > 1
\end{equation}
since $x > 1$. Since the deviation would be strictly profitable, we cannot have $\beta(x)=\beta(x-1)$ in an SE. (This argument implicitly assumes that $\beta(x) \leq x - 2$; if not, the deviation would be unprofitable.)

Step 2. We cannot have two consecutive increases, i.e. $\beta(x)=\beta(x-1) + 1 = \beta(x-2) + 2$.

Suppose (for contradiction) that there were two consecutive increases and consider a bidder with $v = x$. If they reduced their bid by 1, the benefit would be $1$. Meanwhile, the cost would be
\begin{equation} \Delta \mathbb{P} v = x \left( \frac{1}{x+1} \right) <1
\end{equation}
since $x > -1$. Since the deviation would be strictly profitable, we cannot have $\beta(x)=\beta(x-1) + 1 = \beta(x-2) + 2$ in an SE.

Step 3. When cannot have $\beta(x)=\beta(x-1)+1=\beta(x-2)+1$.

Suppose (for contradiction) that this did form part of an SE bidding function and consider a bidder with $v = x-1$. If they increased their bid by $1$, the cost would be $1$. Meanwhile, the benefit would be
\begin{equation}
v \Delta \mathbb{P} \geq x \left( \frac{2}{x+1} \right) > 1
\end{equation}
since $x > 3$. Since the deviation would be strictly profitable, we also cannot have $\beta(x)=\beta(x-1)+1=\beta(x-2)+1$ in an SE.

Now, by Lemmas \ref{lemma 2} and \ref{lemma 3} (no gaps and monotonicity), there are only $4$ ways in which the bidding function can end:

\begin{enumerate}
    \item $\beta(x-2), \beta(x-2), \beta(x-2)$
    \item $\beta(x-2), \beta(x-2), \beta(x-2)+1$
    \item $\beta(x-2), \beta(x-2)+1, \beta(x-2)+1$
    \item $\beta(x-2), \beta(x-2)+1, \beta(x-2)+2$
\end{enumerate}

We can rule out the first and third possibilities since they involve two or more identical bids at the end (Step 1). We can rule out the fourth since it involves two consecutive increases (Step 2). By Step 3, we can rule out the second and final possibility. Since we can rule out all of the possibilities, there is no SE when $n=2$.

Second case: $n \geq 3$. Here the argument is simpler: we simply show that a player with $v=x$ would prefer to win for sure than bid in accordance with their equilibrium strategy.

When $v=x$, following equilibrium yields a payoff of at most
\begin{equation}
x\bigg(\frac{x}{x+1}\bigg)^{n-1} - \beta(x).
\end{equation}
If the player were to bid an extra $1$, they would get
\begin{equation}x-\beta(x)-1.\end{equation}
They are strictly better off winning for sure if
\begin{equation}
x-\beta(x)-1 > x\bigg(\frac{x}{x+1}\bigg)^{n-1} - \beta(x).
\end{equation}
This can be simplified to
\begin{equation}
n-1>\frac{\text{ln}(x)-\text{ln}(x-1)}{\text{ln}(x+1)-\text{ln}(x)} \equiv g(x)
\end{equation}
Assume that $x \geq 2$ (i.e. there are at least 3 possible bids/values). Then it can be shown that $g(x)$ attains its maximum at $x=2$ where $g(x) \approx 1.7$. Thus, if $n \geq 3$ (so $n-1 \geq 2$), it must be that $n-1 > g(x)$. In other words, the bidder with the highest value would strictly prefer to win for sure as opposed to bidding in accordance with an SE.

Next, consider the model with ties. It will be helpful to start with a simple lemma:
\begin{lemma}[No jumps at the top]
If $n=2$ or $n=3$, then $\beta(x-2)=b$ implies $\beta(x-1)\leq b+1$; and similarly, $\beta(x-1)=b$ implies $\beta(x)\leq b+1$.
\end{lemma}
\begin{proof}
Consider initially the case of $n=2$. Suppose (for contradiction) that there is a jump at some value $v$, i.e. $\beta(v)=b$ but $\beta(v+1)=b+m$ where $m\geq 2$. If a player with a value of $v+1$ deviated to bidding $b+m$, they would gain $b+m - (b+1) = m-1 \geq 1$ for sure. However, they would lose (in expectation)
\begin{equation}
(v+1)\Delta \mathbb{P} = (v+1)\frac{\mathbb{P}(b+m)}{2} = (v+1)\frac{\# (b+m)}{2(x+1)}.
\end{equation}
In equilibrium, the loss from deviating cannot exceed the gain of (at least) $1$, yielding
\begin{equation}
\# (b+m) \geq \frac{2(x+1)}{v+1} > 2
\end{equation}
since $v+1 \leq x$. Thus, if there is a jump at a value $v$, it must be that $\# (b+m) > 2$. But then there cannot be a jump when the value is either $x - 1$ or $x - 2$: for in either case, there are fewer than $2$ values strictly greater than $v$, so necessarily $\# (b+m) \leq 2$.

Consider now the case of $n=3$. By a similar argument, we can show that $\# (b+m) \geq 2$ if there is a jump at a value $v$. To see this, suppose (for contradiction) that there is a jump at $v$ but that $\# (b+m) = 1$. As before, the benefit from deviating to a bid of $b+1$ when your value is $v+1$ would be $m-1 \geq 1$. Meanwhile, the expected cost would be
\begin{equation}
\begin{split}
   (v+1)\Delta \mathbb{P}  & =  (v+1)\left(\frac{\mathbb{P}(b+m)^2}{3} + \mathbb{P}(b+m)\mathbb{P}(\leq b)\right) \\
 & = (v+1)\mathbb{P}(b+m)\left(\frac{\mathbb{P}(b+m)}{3} + \mathbb{P}(\leq b)\right) \\
  & = \frac{v+1}{x+1}\left(\frac{1}{3(x+1)} + \mathbb{P}(\leq b)\right) \\
  & \leq \frac{v+1}{x+1}\left(\frac{1}{3(x+1)} + \frac{x}{x+1}\right) < 1
\end{split}
\end{equation}
since $v+1 < 1 + x$. So the cost of deviating would be smaller than the benefit, which implies that this cannot be an SE. Equivalently, if there is a SE with a jump at a value of $v$, it must be that $\# (b+m) \geq 2$.

We now show how this inequality rules out jumps at the top. When $v = x - 1$, the argument is trivial: there is only one valuation strictly greater than $v$, so necessarily $\# (b+m) \leq 1$. Since this contradicts $\# (b+m) \geq 2$, we cannot have a jump at $x - 1$.

Consider now the possibility of a jump when the value is $x - 2$. In any SE with such a jump, $\# (b+m) \geq 2$ and so $\beta(x - 1) = \beta(x) = b + m$. But in that case, a player with a value of $x$ would prefer to win for sure (by bidding one more) than to bid in accordance with equilibrium. To see this, note that such a player's equilibrium payoff would be
\begin{equation}
x \bigg[ \bigg(\frac{x-1}{x+1}\bigg)^2+\frac{1}{3}\bigg(\frac{2}{x+1}\bigg)^2+ \bigg(\frac{2}{x+1}\bigg)\bigg(\frac{x-1}{x+1}\bigg)\bigg]
-\beta(x)
\end{equation}
whereas they would obtain a (certain) payoff of $x - \beta(x) - 1$ by bidding one more. It can be shown that this latter quantity is larger provided that $x\geq 2$. Thus, the deviation would be strictly profitable, implying that there is also no SE with a jump at $x - 2$.\end{proof}

This lemma in hand, we can now demonstrate the non-existence of an SE. Consider initially the case of $n=2$. As before, we will study the last three possible bids and argue that none of the possible sequences can form part of an SE.

Case 1 ($n = 2$): $\beta(x) = \beta(x-1)+1 = \beta(x-2)+2$.

Consider deviating to bidding $\beta(x-1)$ when $v=x$. The gain from doing so would equal $1$. Meanwhile, the (expected) cost would be
\begin{equation}
v \Delta \mathbb{P}  = x \left(\frac{0.5}{x+1}+\frac{0.5}{x+1}\right) = \frac{x}{x+1}<1.
\end{equation}
Since the cost is lower than the gain, this is a strictly profitable deviation --- which means that this sequence of bids cannot form part of an SE.

Case 2 ($n = 2$): $\beta(x) = \beta(x-1) = \beta(x-2)$.

Consider bidding one more when $v=x$. The cost of doing so would equal $1$. Meanwhile, the (expected) benefit would be
\begin{equation}
v \Delta \mathbb{P} \geq x \left( \frac{1.5}{x+1} \right) > 1 \end{equation}
provided that $x \geq 3$. Since the benefit exceeds the cost, this is a strictly profitable deviation, ruling out this sequence of bids in an SE.

Case 3 ($n = 2$): $\beta(x) = \beta(x-1) + 1 = \beta(x-2) + 1$.

Consider bidding one more when $v=x-1$. The cost of doing so equals $1$. The (expected) benefit is
\begin{equation}
v \Delta \mathbb{P} \geq (x - 1) \left( \frac{1}{x+1} + \frac{0.5}{x+1} \right) > 1
\end{equation}
provided that $x\geq6$. Again, this is a strictly profitable deviation, which means that this is not an SE.

Case 4 ($n = 2$): $\beta(x) = \beta(x-1) = \beta(x-2) + 1$.

Consider bidding $1$ more when $v=x-2$. The cost is $1$, whereas the (expected) benefit is
\begin{equation}
v \Delta \mathbb{P} \geq (x-2) \left( \frac{0.5}{x+1} + \frac{1}{x+1} \right) > 1
\end{equation}
provided that $x \geq 9$. So this is (once again) not an SE.

Given monotonicity and no-jumps, this exhausts all the cases, showing that there is no SE with $2$ bidders.

Consider now the case of $n=3$. The argument has a similar structure to before:

Case 1 ($n = 3$): $\beta(x) = \beta(x-1)$.

In that case, a bidder would prefer to bid $1$ more when $v=x$; in fact, this follows directly from the proof of the previous lemma (letting $m = 1$).

Case 2 ($n = 3$): $\beta(x) = \beta(x-1) + 1$.

Consider a bidder with $v=x-1$ who deviates by bidding one more. The cost of doing so is $1$. Meanwhile, the expected benefit is $ v \Delta \mathbb{P}$. Before the deviation,
\begin{equation}
\mathbb{P}\text{(win}) \leq \left(\frac{x-1}{x+1}\right)^2 + \frac{1}{3} \left(\frac{1}{x+1}\right)^2 + \left(\frac{x-1}{x+1}\right)\left(\frac{1}{x+1}\right)
\end{equation}
After the deviation,
\begin{equation}
\mathbb{P}\text{(win}) = \left(\frac{x}{x+1}\right)^2 + \frac{1}{3} \left(\frac{1}{x+1}\right)^2 + \left(\frac{x}{x+1}\right)\left(\frac{1}{x+1}\right)
\end{equation}
A quick calculation thus reveals that
\begin{equation}
v \Delta \mathbb{P} \geq \frac{2vx}{(x+1)^2}   = \frac{2x(x - 1)}{(x+1)^2} > 1
\end{equation}
provided that $x \geq 5$. Thus, the benefits from bidding $1$ more exceeds the costs, confirming that this cannot comprise part of an SE. Furthermore, by monotonicity and no-jumps, these are the only possible cases --- implying that there is no SE at all.

Finally, consider the case of $n \geq 4$. To begin, we find an upper bound on the probability with which a bidder with a value of $x$ (the maximum) can win in equilibrium. In the best case scenario, $\beta(x) > \beta(x - 1)$, i.e. opponents with strictly lower values submit strictly lower bids. In that case, a player with a value of $x$ would win with probability
\begin{equation}
\left(\frac{x}{x+1}\right)^{n-1} + \left(\frac{x}{x+1}\right)^{n-2}\left(\frac{1}{x+1}\right){n-1\choose 1}\left(\frac{1}{2}\right) + .\hspace{0.1em} .\hspace{0.1em} . + \left(\frac{1}{x+1}\right)^{n-1}\frac{1}{n}
\end{equation}
since they win if all opponents have lower values (the first term), if all but one has a lower value and they win the coin flip that decides the tie (the second term, note that this event can happen in $n-1$ ways since they have $n-1$ opponents to tie with), ..., or finally if all their opponents also have the maximum value but they are still chosen to win (the last term). If we define $p \equiv 1/(x + 1)$, we can write the sum $s(n, p)$ more compactly as
\begin{equation}\label{eq 3 again}
s(n, p) = \sum_{i=1}^{n} (1 - p)^{n-i}p^{i - 1}{n-1\choose i - 1}\frac{1}{i} .
\end{equation}
Let us now evaluate (\ref{eq 3 again}) in closed form. As a preliminary, we note that
\begin{align}
{n - 1 \choose i - 1} \frac{1}{i}&= \frac{(n-1)!}{(n-i)!(i-1)!} \frac{1}{i}   =\frac{1}{n} \frac{n!}{(n-i)! \hspace{0.2em}i!} = \frac{1}{n}{n \choose i}.
\end{align}
We can thus write (\ref{eq 3 again}) as
\begin{align}\label{closed form}
s(n, p)
&= \sum_{i=1}^{n} (1 - p)^{n-i}p^{i - 1}\frac{1}{n}{n \choose i} \nonumber
\\&= \frac{1}{np}\sum_{i=1}^{n} (1 - p)^{n-i}p^{i}{n \choose i} \nonumber
\\& = \frac{1}{np}\left(\sum_{i=0}^{n} (1 - p)^{n-i}p^{i}{n \choose i} - (1-p)^n\right) \nonumber
\\& = \frac{1}{np}\bigg([(1 - p) + p]^n- (1-p)^n\bigg) \nonumber
\\& = \frac{1 - (1 - p)^n}{np}
\end{align}
where the penultimate equality follows from the binomial theorem. We claim that $s(n, p)$ is decreasing in $n$. To see this, consider the difference between each term and the next, i.e.
\begin{equation}\label{difference equation}
s(n, p) - s(n + 1, p) = \frac{1 - (1 - p)^n}{np} - \frac{1 - (1 - p)^{n+1}}{(n+1)p} = \frac{1 - (1-p)^n(1 + pn)}{n(n+1)p}.
\end{equation}
Now notice that
\begin{equation}
(1 - p)^{-n} = \sum_{k=0}^{\infty}{n + k - 1 \choose k}p^k = 1 + np + . . .  > 1 + np
\end{equation}
which in turn implies that
\begin{equation}
(1 - p)^{n}(1 + np) < 1
\end{equation}
and so the difference in (\ref{difference equation}) is positive. This means that $s(n, p) > s(n+1, p)$, i.e. our upper bound on the probability of winning falls as the number of bidders increases.

We are now in a position to make the main argument. When $n = 4$, we can use (\ref{closed form}) to bound the equilibrium expected payoff that a player with a value of $x$ can receive. Recalling that $p \equiv 1/(x + 1)$, we see that their payoff cannot exceed
\begin{equation}\label{equilibrium}
x \left(\frac{1 - \left(\frac{x}{x+1}\right)^4 }{4\left(\frac{1}{x+1}\right)}\right) - \beta(x).
\end{equation}
Meanwhile, if they were to bid 1 more, they would win the auction for sure, getting
\begin{equation}\label{deviation}
x - \beta(x) - 1
\end{equation}
It is straightforward but tedious to check that the quantity in (\ref{deviation}) is strictly larger than the quantity in (\ref{equilibrium}) if $x>3.465$. Thus, if $x \geq 4$ (i.e. there are at least 5 values), the deviation is strictly profitable; with the implication that there is no SE.

Finally, recall that $s(n, p)$ is decreasing in the number of bidders $n$. Thus, if the number of bidders rises, we can impose an even tighter upper bound on the equilibrium payoff of a bidder with a value of $x$. In contrast, however, the payoff from the deviation does not change: it is fixed at $x - \beta(x) - 1$. Hence, if the deviation is strictly profitable when $n = 4$, it must also be strictly profitable for arbitrary $n \geq 4$. This concludes the proof.\end{proof}

\begin{proof}[Proof of Proposition \ref{mixed}]
First, we show that, if $n$ is sufficiently large, then all players with values $v \leq x$ must bid $0$ with probability $1$ in any SE. To this end, consider any SE and observe that, for such players, equilibrium expected pay-offs $\pi_{eq}$ can be bounded as follows:
\begin{equation}
\pi_{eq}(v, b) = v \mathbb{P}\text{(win}|b) - b \leq v \left( \frac{x}{x + 1} \right)^{n-1} - b \rightarrow -b \text{  (as } n \rightarrow \infty\text{)}
\end{equation}
where the inequality holds since such types never win against opponents with valuations of $x$. Since expected payoffs become negative for any $b \geq 1$ if $n$ is sufficiently large, this implies that such types must all bid $0$.

Given the logic of lemmas $\ref{lemma 7}-\ref{lemma 10}$, we see that players with a valuation of $x$ must then randomise over an set of consecutive integers whose smallest element is either $0$ or $1$, and whose largest element we denote $k$ (see \cite{rasooly} for details). In fact, the smallest element must be $0$: for it were $1$, then
\begin{equation}
    \pi_{eq}(v = x, b = 1) = v \mathbb{P}\text{(win}|b=1) - 1 = v \left( \frac{x}{x + 1} \right)^{n-1} - 1 \rightarrow -1 \text{  (as } n \rightarrow \infty\text{)}
\end{equation}
and so the $v = x$ type would prefer to deviate to a bid of $0$ (given large enough $n$). Given that the $v = x$ type randomises over the interval $\{0, 1, ..., k \}$, and is therefore indifferent between all bids in this interval, we see that
\begin{equation}
    \pi_{eq}(v = x) = \sum_{i=0}^k \pi_{eq}(v = x|b = i)\mathbb{P}(b = i) = \sum_{i=0}^k \pi_{eq}(v = x|b = 0)\mathbb{P}(b = i) = 0
\end{equation}
But then it must be that $k = x - 1$: if $k \leq x - 2$ (the only other non-dominated possibility), the type could earn a strictly positive expected payoff by deviating to a bid of $k + 1$.\end{proof}

\section{FPSB with fair tie-breaking} \label{sec:appendixb}

We begin by considering the case of $n = 2$ and proving a qualified version of our no-jumps lemma:

\begin{lemma}\label{lemma 13}
Let $n=2$. Then there is no SE with a value $v$ at which $\beta(v-1) = b - k$ where $k \geq 2$ but $\beta(v) = b$ provided that $\beta(v) \geq (2 - \sqrt{3})v \approx 0.27v$.
\end{lemma}

\begin{proof}
Suppose that there is a jump in the SE bidding function, i.e. there exists a $k \geq 2$, $v \in \mathbb{X}$ and $b \in \mathbb{X}$ such that $\beta(v-1) = b - k$ but $\beta(v) = b$. Let $m$ denote the number of values that induce a bid of $b$. In equilibrium, bidding $b$ with a value of $v$ must be at least as good as bidding $b - 1$:
\begin{equation}
\pi(v, b) = (v-b)\left( \frac{v + 0.5m}{S}\right) \geq
\pi(v, b-1) = (v - b + 1)\frac{v}{S}.
\end{equation}
Hence,
\begin{equation}
\begin{split}
m \geq \frac{2v}{v - b}.
\end{split}
\end{equation}
If we write $b = \alpha v$, our inequality becomes
\begin{equation}\label{eq 3}
m \geq \frac{2}{1 - \alpha}.
\end{equation}
Intuitively, the number of $b$ bids ($m$) must be large so that the loss of probability of winning from decreasing one's bid is large --- that way, the deviation is unprofitable.

Now, let $v'$ denote the maximum value that submits a bid of $b$ (so that $v' \equiv v + m - 1$). In equilibrium, bidding $b$ with a value of $v'$ must be at least as good as bidding $b+1$:
\begin{equation}
\pi(v', b) = (v' - b)\left(\frac{v' + 1 -0.5m}{S}\right) \geq \pi(v', b + 1) \geq (v' - b - 1)\left(\frac{v' + 1}{S}\right)
\end{equation}
Rearranging, we obtain
\begin{equation}\label{eq 4}
m \leq \frac{2(v'+1)}{v' - b}.
\end{equation}
Intuitively, the number of $b$ bids ($m$) must be small so that the gain in probability from bidding 1 more is small -- this makes the deviation unprofitable.

To argue that we cannot have jumps an SE, we will claim that (\ref{eq 3}) implies the falsity of (\ref{eq 4}). That way, it is impossible for both inequalities to hold. To do this, we will argue that
\begin{equation}
\frac{2}{1 - \alpha}> \frac{2(v'+1)}{v' - b} \iff v' > v + \frac{1}{\alpha} - 1
\end{equation}
(writing $b = \alpha v$), which holds since
\begin{equation}
v' \geq v + \frac{2}{1-\alpha}
\end{equation}
and
\begin{equation}
v + \frac{2}{1-\alpha} > v + \frac{1}{\alpha}-1
\end{equation}
provided that $\alpha > 2 - \sqrt{3}$.\end{proof}

\begin{proposition}\label{prop 6}
In the first-price auction with two bidders and an odd maximum value, there is an SE in which $\beta(v)=\floor{v/2}$.
\end{proposition}

\begin{proof}
The argument is essentially the same as before. The only difference is that the probability of winning is now the probability that one's bid is strictly higher than the opponent's bid plus the probability that it is the same multiplied by $1/2$.\end{proof}

\begin{proposition}
With two bidders and $x \geq 10$, the only possible SE is $\beta(v)=\floor{v/2}$.
\end{proposition}

\begin{proof}
Since players avoid dominated strategies, $\beta(0) = 0$, $\beta(1) = 0$ and $\beta(2) = 1$. Plainly, then, $\beta(v) = \floor{v/2}$ holds up to (and including) $v = 2$. We will now argue that, if it holds up to any $v\geq 2$, it must hold up to $v+1$. By induction, it then must hold for all $v$.

First case: the value is odd. Let us denote the value by $2v + 1$. Since $\beta(v) = \floor{v/2}$ holds up to the value $2v$, we know that $\beta(2v) = v$ (since $2v$ is even) and $\beta(2v-1) = v-1$ (since $2v-1$ is odd). We wish to demonstrate that $\beta(2v+1) = v$ (i.e. the bidding function does not increase).

To do so, let us suppose (for contradiction) that $\beta(2v+1) \neq v$. Since $\beta(2v+1) \geq (2v+1)/2 \geq (2 - \sqrt{3})v$, we may apply the no-jumps lemma, which implies that $\beta(2v+1) \leq v + 1$. By monotonicity, $\beta(v+1) \geq v$; and so the assumption that $\beta(2v+1) \neq v$ implies that $\beta(2v+1) = v+1$. Now, in equilibrium, a player with a value of $2v+1$ must prefer to bid $v+1$ than to bid only $v$, i.e.
\begin{equation}
\pi(2v+1, v+1) = v \left( \frac{2v+1+0.5\#(v+1)}{S} \right) \geq \pi(2v+1, v) = (v+1) \left( \frac{2v+0.5}{S}\right)
\end{equation}
where $\#(v+1)$ denotes the number of values that induce a bid of $v+1$. This simplifies to
\begin{equation}
\#(v+1) \geq 3 + \frac{1}{v}
\end{equation}
which implies (since $v>0$) that $\#(v+1) > 3$; and therefore (since $\#(v+1)$ is an integer) that $\#(v+1) \geq 4$. Hence, we conclude that there must be at least $4$ valuations that induce a bid of $v+1$.

We now argue that this leads to a contradiction. There are two (sub)cases to consider. In the first case, suppose that $2v+4 \notin \mathbb{X}$ (i.e. the value $2v+1$ is sufficiently close to the end of the bidding function that $2v+4$ exceeds the maximum value). In that case, it is clearly impossible that $\#(v+1) \geq 4$ , so we have our contradiction immediately. In the second case, suppose that $2v+4 \in \mathbb{X}$ and consider the situation when $v=2v+4$. We claim that such a bidder would be made strictly better off by bidding 1 more:
\begin{equation}
\begin{split}
& \pi(2v+4, v+2) \geq  (v+2)\left( \frac{2v+1+\#(v+1)}{S}\right) > \\&  \pi(2v+4, v+1)=(v+3)\left( \frac{2v+1+0.5\#(v+1)}{S}\right)
\end{split}
\end{equation}
which simplifies to
\begin{equation}
\#(v+1) (v+1) > 4v + 2
\end{equation}
which holds since $\#(v+1) \geq 4$. Since this deviation is profitable, there is no SE in which $\beta(2v+1)=v+1$, implying that $\beta(2v+1)=v$. This concludes the proof for the odd case.

Second case: the value is even. Let us denote the value by $2v$ and assume (initially) that $2v \geq 12$, i.e. the value is $6$ or higher. (We will consider the case where $2v < 12$ later.) By the inductive hypothesis, $\beta(v) = \floor{v/2}$ holds up to $2v-1$. In particular, this means that $\beta(2v - 1) = v - 1$, $\beta(2v - 2) = v - 1$, $\beta(2v - 3) = v - 2$ and $\beta(2v - 4) = v - 2$.

We wish to show that $\beta(2v) = v$ (i.e. the bidding function increases). To do this, suppose (for contradiction) that $\beta(2v) \neq v$. By monotonicity, $\beta(2v) \geq v-1 > (2 - \sqrt{3})2v$. Thus, we may again apply the no-jumps lemma; so our assumption that $\beta(2v) \neq v$ implies that $\beta(2v) = v - 1$.

Now, one requirement in equilibrium is that a player with a value of $2v - 3$ cannot strictly benefit by bidding one more than they are supposed to. That is,
\begin{equation}
\pi(2v - 3, v - 2) = (v - 1)\left(\frac{2v - 4 + 1}{S}\right) \geq \pi(2v - 3, v - 1) = (v - 2)\left(\frac{2v - 2 + 0.5\#(v-1) }{S}\right)
\end{equation}
Hence,
\begin{equation}\label{moose}
\begin{split}
   & (v - 1)(2v - 3) \geq (v - 2)(2v - 2 + 0.5\#(v-1)) \\ &\iff \#(v-1) \leq 2\left(\frac{v-1}{v-2}\right)
\end{split}
\end{equation}
Intuitively, the number of values that induce a bid of $v - 1$ must be small: otherwise, the gain from increasing one's bid to $v - 1$ will be sufficiently large that the player will prefer to deviate.

Since $v \geq 6$, inequality (\ref{moose}) implies that
\begin{equation}
\#(v-1) \leq 2\left(\frac{5}{4}\right) = 2.5.
\end{equation}
Plainly, however, this is inconsistent with our assumption that $\beta(2v) = v - 1$: for in that case, we must have $\#(v-1) \geq 3$.

The previous argument establishes that, if $\beta(2v) = \floor{v/2}$ holds up to all even $v \geq 6$, then it must hold for all even $v$. To complete the argument, it remains to check the cases of $v = 4$ and $v = 2$. (When $v = 0$, $\beta(2v) = \floor{v/2}$ holds trivially.)

To begin, consider the case of $v = 2$. By contradiction, assume that $\beta(2) \neq 1$. Clearly, this implies that $\beta(2) = 0$. If this is an equilibrium, a player with a value of $2$ must then prefer to bid $0$ than to bid $1$:
\begin{equation}
\pi(2, 0) = 2\left(\frac{0.5 \#(0)}{S}\right) \geq \pi(2, 1) = \frac{\#(0) + 0.5 \#(1)}{S} \iff \#(1) \leq 0
\end{equation}
Thus, if this is an equilibrium, then there can be no bids equal to $1$; which means that the bids are all $0$ or greater than $2$. By the no-jumps lemma, players then must bid $0$ for all values satisfying
\begin{equation}
2 \geq (2 - \sqrt{3})v \iff v \leq \frac{2}{2 - \sqrt{3}} \approx 7.46
\end{equation}

Hence, $\beta(v) = 0$ for all $v \in \{0, 1, ..., 7\}$. It is then easy to check, however, that a player with a $v = 7$ would then strictly benefit if they deviated to a bid of $1$, contradicting the supposed optimality of equilibrium play.

Finally, consider the case of $v = 4$. By the usual argument, assuming (for contradiction) that $\beta(4) \neq 2$ implies that $\beta(4) = 1$. Furthermore, a player with a value $v = 4$ must prefer their equilibrium bid to bidding one more:
\begin{equation}
\pi(4, 1) = 3\left(\frac{2 + 0.5 \#(1)}{S}\right) \geq \pi(4, 2) = 2\left(\frac{2 + \#(1) + 0.5 \#(2)}{S} \right)
\end{equation}
or equivalently
\begin{equation}\label{quack}
4  \geq \#(1) + 2\#(2)
\end{equation}
Since $\beta(2) = \beta(3) = \beta(4) = 1$, $\#(1) \geq 3$. From (\ref{quack}), we thus conclude that $\#(2) = 0$, i.e. there are no bids of $2$. Applying the no-jumps lemma, we infer that $\beta(v) = 1$ for all $v$ that satisfy
\begin{equation}
3 \geq (2 - \sqrt{3})v \iff v \leq \frac{3}{2 - \sqrt{3}} \approx 11.2
\end{equation}
Thus, $\beta(v) = 1$ for all $v \in \{4, 5, ..., 10\}$. (Since $x \geq 10$, these values are all well defined.) It is easy to check, however, that a player with a value $v = 10$ would strictly benefit by increasing their bid from $b = 1$ to $b = 2$. This concludes the proof for the even case.\end{proof}

\begin{corollary}
If $x \geq 10$, there is no SE in the first-price auction with 2 bidders and an even maximum value.
\end{corollary}

\begin{proof}
From the previous, any SE must have the form $\beta(v)=\floor{v/2}$. However, consider the situation when $v=x$ (the maximum). Equilibrium play calls for $b=0.5x$, so the equilibrium payoff is
\begin{equation}
\frac{x}{2}\bigg(\frac{x}{x+1}+\frac{1}{2(x+1)}\bigg)
\end{equation}
However, if the player deviated to $b=0.5x-1$, they would get
\begin{equation}
\bigg(\frac{x}{2}+1\bigg) \bigg(\frac{x-2}{x+1}+\frac{1}{x+1}\bigg)
\end{equation}
It can be shown that the latter quantity is strictly larger assuming that $x \geq 5$. Thus, there is no SE in which $\beta(v)=\floor{v/2}$ -- and therefore no SE at all.\end{proof}

Finally, we will now briefly summarise our results for the case of $n = 3$ (the calculations are available upon request). Suppose that $x \geq 10$. Then:
\begin{enumerate}
    \item If the $x$ is a multiple of $3$, there is no SE.
    \item Otherwise, there is a SE in which all players set $\beta(v) = (2/3)v$, $\beta(v+1) = (2/3)v$ and $\beta(v+2) = (2/3)v + 1$ for every $v$ that is a multiple of 3.
\end{enumerate}
Like our analysis of the model without ties, this again emphasises the fact that non-existence of SE is common once auctions are made discrete.


\end{appendix}

\end{document}